\documentclass[draftclsnofoot, onecolumn, comsoc, 12pt]{IEEEtran} 

\usepackage{graphicx,epsfig}
\usepackage[noadjust]{cite}
\usepackage{mcite}
\usepackage{amsfonts,helvet}
\usepackage{fancyhdr}
\usepackage{threeparttable}
\usepackage{epsf,epsfig}
\usepackage{amsthm}
\usepackage{amsmath}
\usepackage{siunitx}
\usepackage{amssymb}

\usepackage{dsfont}
\usepackage{subfigure}
\usepackage{color}
\usepackage[linesnumbered,noend,ruled]{algorithm2e}
\usepackage{algpseudocode}
\usepackage{algcompatible}
\usepackage{enumerate}
\usepackage{gensymb}
\usepackage{cancel}
\usepackage{bbm}
\usepackage{graphicx,subfigure}
\usepackage{graphicx}

\newtheorem{lemma}{Lemma}

\newtheorem{proposition}{Proposition}
\newtheorem{remark}{Remark}

\usepackage{eucal}

\def\bff{{\bf f}}

\def\bh{{\bf h}}

\def\bw{{\bf w}}

\def\bF{{\bf F}}

\def\bH{{\bf H}}

\def\bW{{\bf W}}

\begin{document}

\title{Rate-Splitting Multiple Access for Quantized Multiuser MIMO Communications}

\author{Seokjun Park, {\it Student Member, IEEE}, Jinseok Choi,  {\it Member, IEEE}, Jeonghun~Park,  {\it Member, IEEE}, Wonjae Shin, {\it Senior Member, IEEE},  and \\ Bruno Clerckx, {\it Fellow, IEEE} 

\thanks{
S. Park and J. Choi are with Department of Electrical Engineering, Ulsan National Institute of Science and Technology, South Korea (e-mail: {\texttt{\{seokjunpark, jinseokchoi\}@unist.ac.kr}}).

J. Park is with the School of Electronics Engineering, Kyungpook National University, South Korea (e-mail: {\texttt{jeonghun.park@knu.ac.kr}}). 

W. Shin is with the Department of Electrical and Computer Engineering, Ajou University, Suwon 16499, South Korea (email: {\texttt{wjshin@ajou.ac.kr}}).

B. Clerckx is with the Communications and Signal Processing Group, Department of Electrical and Electronic Engineering, Imperial College London, London SW7 2AZ, U.K. (email: {\texttt{b.clerckx@imperial.ac.uk}}).
}

}
\maketitle \setcounter{page}{1} 
\begin{abstract}
This paper investigates the sum spectral efficiency maximization problem in downlink multiuser multiple-input multiple-output (MIMO) systems with low-resolution quantizers  at an access point (AP) and users.
In particular, we consider rate-splitting multiple access (RSMA) to enhance spectral efficiency by offering opportunities to boost achievable degrees of freedom.
Optimizing RSMA precoders, however, is highly challenging  due to the minimum rate constraint when determining the rate of the common stream.
The quantization errors coupled with the precoders further make the problem more complicated and difficult to solve.
In this paper, we develop a novel RSMA precoding algorithm incorporating quantization errors for maximizing the sum spectral efficiency.
To this end, we first obtain an approximate spectral efficiency in a smooth function. 
Subsequently, we derive the first-order optimality condition in the form of the nonlinear eigenvalue
problem (NEP). 
We propose a computationally efficient algorithm to find the principal eigenvector of the NEP as a sub-optimal solution.
Simulation results validate the superior spectral efficiency of the proposed method.
The key benefit of using RSMA over spatial division multiple access (SDMA) comes from the ability of the common stream to balance between the channel gain and quantization error in multiuser MIMO systems with different quantization resolutions.
\end{abstract}

\section{Introduction}

Nowadays, 6G wireless communication has drawn significant attention beyond the era of 5G \cite{yang20196g}.
Accordingly, low-power yet high-speed wireless communications have become more indispensable;
there are applications such as the internet-of-things (IoT) in which devices  tend to be battery-limited and to have low computing capability while requiring high spectral efficiency  \cite{pattar2018searching, minoli2017iot}.
We can alleviate the power consumption issue by utilizing low-power hardware, such as low-resolution analog-to-digital converter (ADC) and digital-to-analog converter (DAC), since the power consumption of the quantizers decreases exponentially according to the reduction of quantization bits \cite{zhang2018low}.
Motivated by the power saving in using low-resolution quantizers, low-resolution quantization systems or mixed-resolution quantization systems where high- and low-resolution quantizers coexist have been widely studied
\cite{zhang2016mixed,liang2016mixed,zhang2018mixed,choi:commmag:20}. 

Another key challenge needed to be addressed for the realization of 6G wireless communications is the severe inter-user interference due to the dramatic increase in the number of smart devices.
In IoT communications, for example, the increase in the number of IoT devices and the high channel correlation among IoT devices \cite{pattar2018searching, minoli2017iot} can incur a significant amount of inter-user interference, thereby leading to considerable performance degradation in spectral efficiency.
In this regard, rate-splitting multiple access (RSMA) introduced in \cite{mao2018rate_bri} for downlink multi-antenna wireless networks as a unified multiple access scheme can be an effective solution in overcoming the limits of the spectral efficiency gain in  multiuser multiple-input multiple-output (MU-MIMO) systems by reducing the inter-user interference \cite{joudeh:16:tsp, dai:twc:16, li:jsac:20,clerckx2016rate}.
In this paper, we consider RSMA for downlink MU-MIMO systems with low-resolution quantizers and develop a novel and computationally efficient precoding method to maximize the spectral efficiency.
\subsection{Prior Work}
The information-theoretic performance limits when using low-resolution DACs and ADCs have been widely studied in the literature. 
The capacity of uniformly distributed quadrature phase shift keying (QPSK) was achieved in a quantized single-input single-output (SISO) additive white Gaussian noise (AWGN) channel \cite{singh2009limits}.
In addition, the capacity of multiple-input single-output (MISO) fading channel with one-bit ADCs was derived in a closed form \cite{mo2015capacity}.
The analytical performance of low-resolution quantization systems was further studied via  linear approximation of quantization process, namely,  Bussgang decomposition  \cite{mezghani2012capacity} and additive quantization noise model (AQNM) \cite{orhan2015low}.
In \cite{mezghani2012capacity}, the lower bound on the achievable rate of the quantized MIMO channel was derived based on the Bussgang decomposition.
In addition, in \cite{orhan2015low}, optimal bandwidth and resolution of ADC for the SISO channel were analyzed employing the AQNM.


The existing conventional precoding methods for a perfect quantization system  revealed highly limited  spectral efficiency due to the non-negligible quantization error which was not taken into account~\cite{jacobsson2017quantized,choi2021energyIOTJ}.
Accordingly, in \cite{jacobsson2017quantized},  conventional  precoding methods such as minimum mean square error (MMSE) or zero-forcing (ZF) for 3 to 4-bit DACs were proposed.
In particular, the proposed precoders with 3 to 4-bit DACs achieved comparable performance with the  high-resolution DAC system.
The alternating direction method of multipliers (ADMM) was also used to solve the inter-user interference minimization problem in massive MU-MIMO systems with low-resolution DACs \cite{wang2018finite}.
To design the precoder for  downlink MU-MIMO systems with heterogeneous-resolution DACs and ADCs where the DACs (also ADCs) have different resolutions from each other, a generalized power iteration-based algorithm was proposed to maximize the energy efficiency \cite{choi2021energyIOTJ}.
Mixed DACs and ADCs architectures which are the special cases of the heterogeneous-resolution DACs and ADCs  were further investigated, revealing the potentials in increasing the spectral efficiency compared to the case of homogeneous-resolution DACs and ADCs  where  the DACs (also ADCs) have the same resolutions to each other \cite{zhang2016mixed, pirzadeh2018spectral}.

One-bit quantization systems have also been investigated to develop state-of-the-art channel estimation and detection techniques because of the implementation practicality  and analytical tractability of one-bit quantizers \cite{choi2016near, li2017channel,  jeon2018supervised, choi2019robust}.
Channel estimation or detection techniques for a one-bit quantization system were developed by using a maximum likelihood detector \cite{choi2016near} and Bussgang decomposition \cite{li2017channel}.
In addition, learning-based detection methods without the necessity of explicit channel estimation were proposed in  \cite{jeon2018supervised, choi2019robust}.
Although the prior precoding methods for low-resolution quantization systems achieved high improvement in spectral efficiency, they are still limited to a  conventional signaling method that may not be optimal in the system with severe inter-user interference.
As such, we  also consider an advanced signaling method for low-resolution quantization systems to further increase  spectral efficiency by managing the inter-user interference while reducing the power consumption at transceivers.

RSMA introduced in \cite{joudeh2016sum} theoretically showed its optimality in achieving the channel degree-of-freedom with imperfect channel state information (CSI) by reducing interference.
RSMA can be considered as a generalization of the  Han-Kobayash scheme \cite{han:tit:81} for broadcast channels which guarantees the achievable rate within one bit per second per hertz of the  capacity in a Gaussian interference channel model  \cite{etkin:tit:08}.
The key principle of RSMA is to divide  each user's stream into a common stream and a private stream.
The codebook of the common stream is shared with all users so that they can eliminate the common stream using successive interference cancellation (SIC) when decoding their private streams.
Hence, the private streams undergo a reduced amount of interference, thereby improving the  spectral efficiency \cite{mao2018rate_bri, joudeh2016sum}.

Motivated by \cite{yang:13:tit}, the two receivers' achievable sum spectral efficiency in the MISO channel was analyzed with the randomly generated precoder for the common stream and ZF precoder for the private streams \cite{hao:tcom:15}. 
To maximize the spectral efficiency in RSMA, a linear precoding method in the downlink MISO system  based on weighted MMSE (WMMSE) \cite{chris:twc:08} was proposed in \cite{joudeh2016sum}.
In \cite{joudeh2016sum}, the spectral efficiency maximization problem was converted to a quadratically constrained quadratic program (QCQP) by minimizing the mean square error 
 (MSE).
In addition, the precoder design of RSMA was represented by approximating the optimization problem in convex form with the concave-convex procedure (CCCP) \cite{li:jsac:20}.
A hierarchical RSMA architecture where  more than one common stream is decodable depending on the hierarchy was proposed in downlink massive MIMO system \cite{dai:twc:16}.
A generalized RSMA framework was proposed in \cite{mao2018rate_bri} to bridge, generalize and outperform existing multiple access techniques such as orthogonal multiple access (OMA), non-orthogonal multiple access (NOMA), spatial-division multiple access (SDMA), and multiuser MIMO.
The optimal rate allocation and power control algorithm of the common  and  private streams was proposed in \cite{yang:icc:20}.
Furthermore, in \cite{park2021rate}, the RSMA precoder design algorithm which is based on the generalized power iteration (GPI) algorithm  was proposed to maximize the sum spectral efficiency in the downlink MU-MIMO system considering the channel estimation error.

RSMA has emerged as a promising enabler in a wide range of  applications \cite{mao2022rate}.
The performance analysis of RSMA in low-resolution quantization  systems with regularized zero-forcing (RZF) was presented in \cite{ahiadormey2021performance}.
In \cite{dizdar2021rate}, a precoder design algorithm based on the ADMM of multi-antenna joint radar-communication (JRC) system in RSMA was proposed to solve the joint sum rate maximization problem considering low-resolution DACs.
Furthermore, in \cite{dizdar2022energy}, ADMM based algorithm for energy-efficient dual-functional radar-communication in RSMA was proposed with low-resolution DACs and RF chain selection.


We remark that the prior works of RSMA successfully showed the benefits of RSMA and developed state-of-the-art RSMA transmission methods.
However, a rigorous investigation on a precoding design problem for downlink RSMA  considering the general  number of DAC and ADC bits is still missing.
Although in \cite{dizdar2021rate, dizdar2022energy} the spectral and energy efficiencies of RSMA were studied for low-resolution DACs, the derived results are   constrained to systems with  homogeneous DACs at the transmitter and perfect ADCs at the receiver.
In addition, the proposed methods in \cite{dizdar2021rate, dizdar2022energy} are optimized for the joint radar and communications systems, which may lead to sub-optimal performance when considering communication only.
More importantly, to compute the spectral efficiency in RSMA, it is necessary to consider the minimum rate of the common stream which is a non-smooth function.
In most of the literature, such a non-smooth optimization problem induced by the minimum rate condition was often solved using the convex relaxation method  based on the CVX toolbox.
The existing optimization methods  which are based on the CVX  \cite{dizdar2021rate, dizdar2022energy} have high computational complexity. 
In this regard, developing a computationally efficient precoding method with improved performance is necessary to maximize the sum spectral efficiency for downlink  RSMA  systems with heterogeneous DACs and ADCs.


\subsection{Contributions}
In this paper, we propose a precoding optimization framework and study the effect of RSMA under coarse quantization systems. 
We summarize our contributions as follows:
\begin{itemize}
    \item We consider an RSMA system in which a multi-antenna AP transmits a common stream and private streams via linear precoding to  single-antenna users.
The AP is equipped with low-resolution DACs of arbitrary resolutions, and users are also equipped with low-resolution ADCs of arbitrary resolutions. 
    This starkly contrasts with \cite{joudeh2016sum,dizdar2021rate, dizdar2022energy} where no quantization error or only DAC quantization error was considered with homogeneous DACs.
    We then  formulate the sum spectral efficiency optimization problem by incorporating the quantization errors from both the DACs and ADCs and considering the rate of the common stream.
    Since the common stream needs to be decodable by all users in RSMA, the rate of the common stream becomes the minimum rate supportable by all users' channels.

    \item We propose a novel and computationally efficient precoding method for the considered RSMA system.
    To this end, 
    we resolve the challenge that arises from the non-smoothness first by  deriving the lower bound of the sum spectral efficiency utilizing the LogSumExp approximation technique to convert the non-smooth minimum function into a tractable form.
    Since the minimum rate condition for the common stream was often handled by introducing an inequality constraint for each user's common stream \cite{joudeh2016sum,dizdar2021rate, dizdar2022energy},  such an approximation contributes to reducing algorithm complexity.
    Then, reformulating the sum spectral efficiency further, we establish the first-order optimality condition 
    to identify stationary points.
    We interpret the derived condition as a nonlinear eigenvalue problem (NEP) \cite{cai:siam:18} in which the eigenvalue and eigenvector correspond to the sum spectral efficiency and precoding vector, respectively.
    Hence, finding the leading eigenvector is equivalent to finding the best local optimal point.
    Based on the insight, we propose a quantized generalized power iteration for rate-splitting (Q-GPI-RS) algorithm by leveraging that a power iteration efficiently seeks the principal eigenvector. 
    Consequently, using the approximation of the minimum function and power iteration, the proposed Q-GPI-RS precoding method has low complexity order compared to the CVX-based methods.
  

    \item Extensive numerical results show that Q-GPI-RS outperforms the conventional linear precoding methods such as regularized ZF (RZF), ZF, and maximum ratio transmission (MRT) in terms of spectral efficiency.
   To compare the performance of  Q-GPI-RS method with the existing RSMA precoding algorithm, we also introduce the quantization-aware WMMSE-based alternating optimization (Q-WMMSE-AO) which is the extended version of the existing RSMA precoding to the considered quantization systems.
    Although Q-WMMSE-AO improves the spectral efficiency and shows higher robustness compared to the conventional WMMSE-AO approach \cite{joudeh2016sum}, Q-GPI-RS method largely outperforms Q-WMMSE-AO in most environments.
    In addition, compared to spatial division multiple access (SDMA),  RSMA offers a noticeable spectral efficiency gain in the medium to high signal-to-noise (SNR) regime where the performance is limited by the inter-user interference and quantization error. 
    Therefore, the simulation results validate not only the effectiveness of the proposed method, but also the benefit of RSMA over SDMA.
    \item Key findings from the numerical results are:
    $(i)$ As shown analytically in \cite{clerckx2019rate}, RSMA has been known to offer larger improvement of spectral efficiency with higher channel correlation,
     and we  confirm  such benefit of RSMA also  in the presence of DAC and ADC quantization errors.
    $(ii)$ Considering homogeneous DACs and ADCs, the spectral efficiency gain of RSMA increases with the resolutions due to the fact that quantization errors involved with the common stream cannot be canceled at the users.
    In addition, by reformulating the signal-to-interference-plus-noise ratios (SINRs), we estimate that the effect of ADC is more dominant than DAC in deciding the rate of the common stream, which is  verified by numerical results.
   This finding is not observed in previous works \cite{dizdar2021rate, dizdar2022energy} as  only the effect of DAC quantization error was considered.
   The finding indicates the importance of considering ADC quantization errors as well as DAC quantization errors in developing a transmission strategy.
    $(iii)$ In the heterogeneous DAC system, mixed DAC systems, in particular, antennas with low-resolution DACs tend to be turned off in the high SNR to reduce the quantization error for RSMA.
    This phenomenon may lead to overloaded systems due to the insufficient number of active antennas, and thus RSMA can play a key role in maximizing the sum spectral efficiency by leveraging the common stream; however, SDMA tends to use both the high- and low-resolution DACs to avoid the overloaded system while  suffering from the large quantization error. 
\end{itemize}

\textit{Notation}:
$\text{a}$ is a scalar, ${\bf{a}}$ is a vector and ${\bf{A}}$ is a matrix.
The superscripts $(\cdot)^{\sf T}$, $(\cdot)^{\sf H}$, and $(\cdot)^{-1}$ denote matrix transpose, Hermitian, and inversion, respectively.
$\mathbb{E}[\cdot]$ and $\text{tr}(\cdot)$ represent  expectation operation and trace of a matrix, respectively.
${\bf{I}}_K$ is the identity matrix of size $K \times K$.
${\bf 0}_{N}$ is the zero matrix of size $N \times N$ and ${\bf 0}_{N \times 1}$ is the zero vector of size $N \times 1$.
${\bf{A}} = {\rm blkdiag}\left({\bf{A}}_1, ...,{\bf{A}}_n,..., {\bf{A}}_N \right)$ is a block diagonal matrix with block diagonal entries  of ${\bf{A}}_1, ..., {\bf{A}}_N$. 

\section{System Model}
\label{sec:sys_model}

We consider a downlink MU-MIMO system, in which an access point (AP) with $N$ antennas serves $K$ single-antenna  users as shown in Fig.~\ref{fig:RS-SE-DAC-ADC system model}.
The AP employs DACs with $b_{{\sf DAC},n}$-bit resolution where $b_{{\sf DAC},n}$ represents the number of quantization bits of the DAC pair at $n$-th antenna.
Each user employs an ADC of $b_{{\sf ADC},k}$-bit resolution where $b_{{\sf ADC},k}$ represents the number of quantization bits of the ADC pair at $k$-th user.

We now explain the process of RSMA shown in Fig.~\ref{fig:RS-SE-DAC-ADC system model}.
We use the 1-layer rate-splitting (RS) architecture of RSMA \cite{clerckx2016rate,joudeh2016sum}.
Let $W_1, W_2,\dots, W_K$ be
the uncoded messages for users $1,2,\dots,K$, respectively.
Considering 1-layer RSMA, each message is
divided into two parts, i.e., $W_{k0}$ and $W_{k1}$, $k=1,\dots,K$, which correspond to the common part and private part of $W_k$ respectively. 
All common parts are combined into one super common message, i.e. $W_0 = (W_{10},\dots, W_{K0})$. 
The resulting $K+1$ messages of ($W_0, W_{11},\dots, W_{K1}$) are  encoded into symbol streams:  $K$ private streams $s_k$ and a  common 
stream $s_{\sf c}$.
The common {stream} $s_{\sf c}$ is coded from a public codebook, and thus RSMA users can decode the common {stream}.
Accordingly, the common {stream} is decoded first and  eliminated when decoding private {streams} via SIC.
We assume $s_{\sf c}, s_k \sim {\mathcal {CN} (0,1)}$, i.e.,  a complex Gaussian distribution with zero mean and unit variance.

\begin{figure}[!t]\centering
	\subfigure{\resizebox{0.9\columnwidth}{!}{\includegraphics{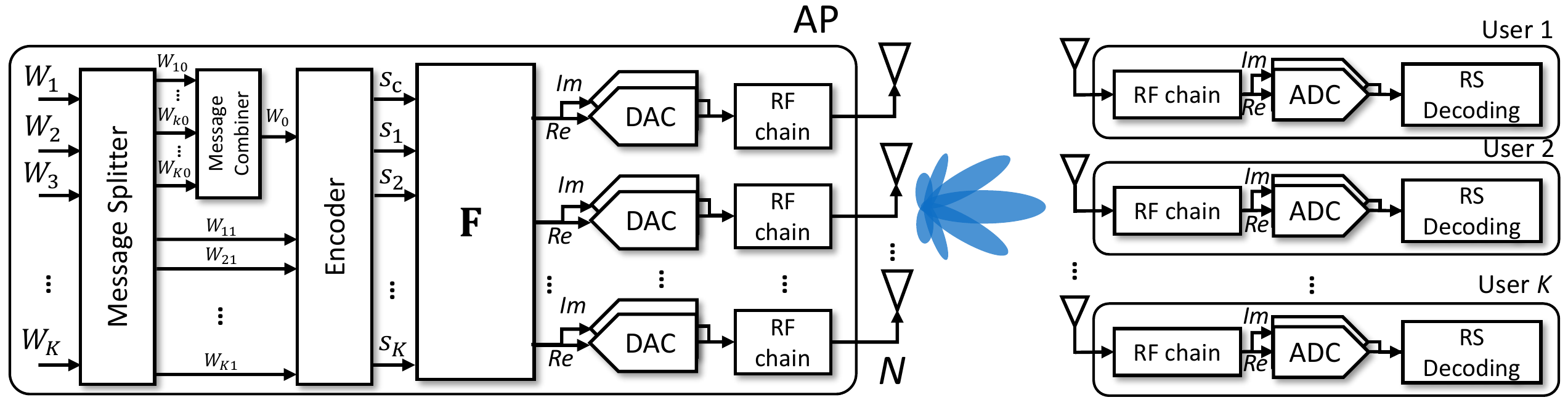}}}
	\vspace{-1em}
	\caption{A multiuser MIMO  downlink  system with rate-splitting multiple access in the low-resolution quantization regime.} 
 	\label{fig:RS-SE-DAC-ADC system model}
 	\vspace{-1em}
\end{figure}

The common stream and private streams are linearly precoded at the AP.
The digital baseband signal ${\bf{x}} \in \mathbb{C}^{N}$ is given by
\begin{align}
    {\bf{x}} = \sqrt{P}{\bf{f}}_{0} s_{\sf c} + \sqrt{P}\sum_{k = 1}^{K} {\bf{f}}_k s_k,
\end{align}
where ${\bf f}_{0} \in \mathbb{C}^{N}$ and ${\bf f}_k \in \mathbb{C}^{N}$ are precoding vectors for the common and private {streams}, respectively, and $P$ is the maximum transmit power.
Then, $\bf{x}$ is quantized at the DACs.
We adopt the AQNM method \cite{fletcher2007robust} that approximately models the quantization process in linear form.
After applying the AQNM, the quantized signal is represented as 
\begin{align} 
    {Q}({\bf{x}})  \approx {\bf x_{\sf q}} = \sqrt{P}{\bf{\Phi}}_{{\alpha }_{\sf DAC}} {\bf{f}}_{0} s_{\sf c} + \sqrt{P}{\bf{\Phi}}_{{\alpha }_{\sf DAC}}\sum_{k = 1}^{K} {\bf{f}}_k s_k + {\bf q}_{\sf DAC},
\end{align}
where $Q(\cdot)$ is a scalar quantizer which applies for each real and imaginary part, ${\bf{\Phi}}_{{\alpha }_{\sf DAC}} = {\rm diag}({\alpha }_{{\sf DAC},1},\ldots,{\alpha }_{{\sf DAC},N})\in \mathbb{C}^{N \times N}$ denotes a diagonal matrix of quantization loss, and ${\bf q }_{\sf DAC}\in \mathbb{C}^{N}$ is a DAC quantization noise vector.
The quantization loss of the \textit{n}-th DAC ${ \alpha }_{\sf DAC,\textit{n}}\in \mathbb (0,1)$ is 
determined as ${ \alpha }_{\sf DAC,\textit{n}} = 1 - { \beta }_{\sf DAC,\textit{n}}$, where ${ \beta }_{\sf DAC,\textit{n}}$ is a normalized mean squared quantization error ${ \beta }_{\sf DAC,\textit{n}} = \frac {\mathbb{E}{[| {x-Q_n(x)}|]^2}} {\mathbb{E}{[|{x}|]^2}}$ \cite{fletcher2007robust,zhang2018mixed}.
The values of ${ \beta }_{\sf DAC,\textit{n}}$ depend on the number of quantization bits $b_{{\sf DAC},n}$.
Especially, if the value of $b_{{\sf DAC},n}$ is less than 5, ${ \beta }_{\sf DAC,\textit{n}}$ is represented in Table 1 in \cite{fan2015uplink}.
If the value of $b_{{\sf DAC},n}$ is larger than 5, ${ \beta }_{\sf DAC,\textit{n}}$ can be approximated as $\frac {\pi \sqrt{3}}2 {2^{-2b_{{\sf DAC},n}}}$ \cite{fan2015uplink}.
The quantization noise is not correlated with digital baseband signal $\bf{x}$ and considered to follow ${{\bf q}_{\sf DAC}}  \backsim {\mathcal {CN} ({\bf 0}_{N\times1},\bf R_{{\bf q}_{\sf DAC}{\bf q}_{\sf DAC}})}$  which is the worst case in terms of spectral efficiency.
Let ${\bf F} = \left[{\bf f}_{0},{\bf{f}}_1, \cdots, {\bf{f}}_K\right]\in \mathbb{C}^{N\times(K+1)}$.
Then, the covariance matrix of ${{\bf q}_{\sf DAC}}$ is computed as \cite{fletcher2007robust}
\begin{align}
    \label{eq:covariance of quantization noise}
    {\bf R_{{\bf q}_{\sf DAC}{\bf q}_{\sf DAC}}}  = {\bf{\Phi}}_{{\alpha }_{\sf DAC}}{\bf{\Phi}}_{{\beta }_{\sf DAC}}{\rm diag} \left({\mathbb E}\left[\bf x\bf x^{\sf H}\right]\right).
\end{align}
Then, the received analog baseband signal vector at $K$ users is
\begin{align} \label{eq:received signal}
    {\bf y} = {\bf H}^{\sf H}{\bf x}_{\sf q} + {\bf n},
\end{align}
where ${\bf H}^{\sf H} \in \mathbb {C}^{K \times N}$ is a downlink channel matrix between the AP and $K$ users, and ${{\bf n}} \backsim {\mathcal {CN} ({\bf 0}_{K\times1},\sigma^{2}{\bf I}_{K})}$ is an AWGN with zero mean and variance of ${\sigma^2}$.
We assume that both the AP and users have the perfect knowledge of channel state information (CSI).\footnote{
Although ADC and DAC quantization on the channel sounding (along with other impairments) would lead to imperfect CSI knowledge, we shall leave the impact of such impairments for future studies.}
At each user,  the received analog signal $y_k$ is quantized by the ADCs of $b_{{\sf ADC},k}$ bits.
Then, the received digital baseband signals are represented as \cite{fletcher2007robust,choi2018spatial}
\begin{align}
    {Q}({\bf{y}})&\approx {\bf y_{\sf q}} = {\bf{\Phi}}_{{\alpha}_{\sf ADC}}{\bf y} + {\bf q_{\sf ADC}}
    \\ 
    \nonumber
    &= \sqrt{P}{\bf{\Phi}}_{{\alpha }_{\sf ADC}}{\bf H}^{\sf H}{\bf{\Phi}}_{{\alpha }_{\sf DAC}}{\bf{f}}_{0} s_{\sf c} + \sqrt{P}{\bf{\Phi}}_{{\alpha }_{\sf ADC}}{\bf H}^{\sf H}{\bf{\Phi}}_{{\alpha }_{\sf DAC}}\sum_{k = 1}^{K} {\bf{f}}_k s_k+ {\bf{\Phi}}_{{\alpha }_{\sf ADC}}{\bf H}^{\sf H}{\bf q }_{\sf DAC}+{\bf{\Phi}}_{{\alpha}_{\sf ADC}}{\bf n}+{\bf q_{\sf ADC}},
\end{align}
where ${\bf{\Phi}}_{{\alpha}_{\sf ADC}} = {\rm diag}({\alpha }_{{\sf ADC},{\sf 1}},\ldots,{\alpha }_{{\sf ADC},{K}})\in \mathbb{C}^{K\times K}$ denotes a diagonal matrix of quantization loss, and ${\bf q }_{\sf ADC}\in \mathbb{C}^{K}$ is a ADC quantization noise vector.
The quantization loss of the $k$-th ADC ${ \alpha }_{{\sf ADC},k}$ and ${ \beta }_{{\sf ADC},k}$ is defined as likewise quantization loss of the DAC.
The ADC quantization noise ${\bf{q}}_{\sf ADC}$ is not correlated with analog baseband signal $\bf{y}$ and considered to follow ${{\bf q}_{\sf ADC}}  \backsim {\mathcal {CN} ({\bf 0}_{K\times1},\bf R_{{\bf q}_{\sf ADC}{\bf q}_{\sf ADC}})}$ which is the worst case in terms of spectral efficiency.
The covariance is derived as
\begin{align}
    \label{eq:RqADC}
    \bf R_{{\bf q}_{\sf ADC}{\bf q}_{\sf ADC}} = {\bf{\Phi}}_{{\alpha }_{\sf ADC}}{\bf{\Phi}}_{{\beta }_{\sf ADC}}{\rm diag} \left({\mathbb E}\left[\bf y\bf y^{\sf H}\right]\right).
\end{align}
Accordingly, the digital baseband signal at user \textit{k} is represented as
\begin{align}
    y_{{\sf q},k}
    =& \sqrt{P}{\alpha }_{{\sf ADC},k}{\bf h}_k^{\sf H}{\bf{\Phi}}_{{\alpha }_{\sf DAC}}{\bf{f}}_{0} s_{\sf c} + \sqrt{P}{\alpha }_{{\sf ADC},k}{\bf h}_k^{\sf H}{\bf{\Phi}}_{{\alpha }_{\sf DAC}}{\bf{f}}_{k} s_{k} 
    \nonumber
    \\
    \label{eq:yqk}
    &+ \sqrt{P}{\alpha }_{{\sf ADC},k}\sum_{i =1, i \neq k}^{K} {\bf h}_k^{\sf H}{\bf{\Phi}}_{{\alpha }_{\sf DAC}}{\bf{f}}_i s_i+{\alpha }_{{\sf ADC},k}{\bf h}_k^{\sf H}{\bf q }_{\sf DAC}+{\alpha }_{{\sf ADC},k}{n}_k\!+\!{q}_{{\sf ADC},k},
\end{align}
where ${\bf h}_k$ is the $k$-th column of the channel matrix $\bf H$.
We remark that  both the ADC and DAC quantization error terms, ${\alpha }_{{\sf ADC},k}{\bf h}_k^{\sf H}{\bf q }_{\sf DAC}$ and  $q_{{\sf ADC},k}$,  in \eqref{eq:yqk} involve the common stream as well as the private streams and corresponding precoders.
Hence, it is expected that without properly designing the precoder for the common stream it is more difficult to accomplish the potential of RSMA in the considered low-resolution quantization systems than the systems with perfect quantization or low-resolution DACs. 

\subsection{Performance Metrics and Problem Formulation}
Recall the decoding principle of RSMA \cite{joudeh2016sum}; each user  decodes the common {stream} $s_{\sf c}$ by treating the private streams as noise.
After $s_{\sf c}$ is successfully decoded, each user can cancel the common {stream} from the received signal $\bf{y}$ by using SIC.
In this regard, the combination of message split and SIC enables to partially decode interference and treat the remaining interference as noise and consequently increases the sum spectral efficiency of users \cite{mao2018rate_bri}.
%
To successfully perform SIC, the common stream needs to be designed carefully in a way that it is decodable to all RMSA users.
Accordingly, the rate of the common stream $s_{\sf c}$ is determined by the minimum of the spectral efficiencies of all users.
Then, the spectral efficiency of $s_{\sf c}$ is defined as
\begin{align}
    \label{eq:ergodic common message}
    R_{\sf c} = \min_{k \in \CMcal{K}} \left\{  \log_2 \left(1 + \frac{P {\alpha^2_{{\sf ADC},k}}|{\bf h}_k^{\sf H}{\bf{\Phi}}_{{\alpha }_{\sf DAC}}{\bf{f}}_{0}|^2} {\text{IUI}_{\text c}+\text{QE}_{k}+{\alpha}_{{\sf ADC},k}^2{\sigma}^2} \right)\right\}
    =\min_{k \in \CMcal{K}} \left\{R_{{\sf c},k}\right\},
\end{align}
where 
\begin{align}
    \label{eq:QEk}
    \text{IUI}_{\text c} = P {\alpha^2_ {{\sf ADC},k}}{\sum_{k = 1}^{K} |{\bf{h}}_{k}^{\sf H}{\bf{\Phi}}_{{\alpha }_{\sf DAC}} {\bf{f}}_{k}|^2} \text{ and }
    \text{QE}_{k} = {\alpha }_{{\sf ADC},k}^2{\bf h}_k^{\sf H}{\bf R}_{{\bf q}_{\sf DAC}{\bf q}_{\sf DAC}}{\bf h}_k + r_{{\sf q}_{{\sf ADC},k}{\sf q}_{{\sf ADC},k}}.
\end{align}
After decoding and eliminating the common stream using SIC, the achievable spectral efficiency of the private stream $s_k$ of  user $k$ is formulated as
\begin{align}
    \label{eq:ergodic private message}
    R_{k} &= \ \log_2 \left(1 + \frac{P {\alpha^2_{{\sf ADC},k}}|{\bf h}_k^{\sf H}{\bf{\Phi}}_{{\alpha }_{\sf DAC}}{\bf{f}}_{ k}|^2} {\text{IUI}_{k}+\text{QE}_{k}+{\alpha }_{{\sf ADC},k}^2{\sigma}^2} \right)\,.
\end{align}
where $\text{IUI}_{k} = P {\alpha^2_{{\sf ADC}, k}}{\sum_{i = 1,i \neq k}^{K} |{\bf{h}}_{k}^{\sf H}{\bf{\Phi}}_{{\alpha }_{\sf DAC}} {\bf{f}}_{ i}|^2}.$
\begin{remark}[SIC and error propagation]\normalfont
We note that when performing SIC users can only remove the quantized common stream, i.e., $ \sqrt{P}{\alpha }_{{\sf ADC},k}{\bf h}_k^{\sf H}{\bf{\Phi}}_{{\alpha }_{\sf DAC}}{\bf{f}}_{0} s_{\sf c}$,  not the quantization errors that come from the common stream.
Under the AQNM, we consider that the additive quantization error follows a Gaussian distribution which is the worst case in terms of the spectral efficiency. 
Accordingly, the derived common rate is considered to be conservative. 
In this regard, we assume perfect decoding of the common stream if the derived rate of the common stream is less than or equal to the minimum of the common rates and thus, no SIC error propagation occurs. 
\end{remark}

We note that the user signals are quantized at the DACs and ADCs, and the quantization noise error terms appear in the rates of the common and private streams.
To maximize the sum spectral efficiency, an optimization problem is formulated as
\begin{align}
    \label{eq:problem_main}
    \mathop{{\text{maximize}}}_{{\bf{f}}_{0}, {\bf{f}}_1, \cdots,{\bf{f}}_K}& \;\;  R_{{\sf c}}+{\sum_{k=1}^{K}}R_k =  R_{\sum}\\
    \label{eq:problem_main_constraint}
    \mathop{{\text{subject to}}}& \;\;{\text {tr}}\left( \mathbb{E}\left[{\bf x}_{\sf q}{\bf x}_{\sf q}^{\sf H}\right]\right) \leq P,
\end{align}
where \eqref{eq:problem_main_constraint} is the  transmit power constraint.
In the following section, we propose a novel computationally efficient precoding method to solve \eqref{eq:problem_main}. 

\section{Precoder Optimization Using Generalized Power Iteration} \label{sec:main}
Since the direct solution of the problem in \eqref{eq:problem_main} is not available due to the non-convexity and non-smoothness of the problem, we present key techniques to convert the formulated problem to a tractable form and solve the reformulated problem by identifying a stationary point  with maximum spectral efficiency.
\subsection{Problem Reformulation}
We first simplify the  constraint in \eqref{eq:problem_main_constraint}.
The covariance matrix of ${\bf q}_{\sf DAC}$ in \eqref{eq:covariance of quantization noise} is derived as
\begin{align}
    {\bf R_{{\bf q}_{\sf DAC}{\bf q}_{\sf DAC}}}
     = &  \;{\bf{\Phi}}_{{\alpha }_{\sf DAC}}{\bf{\Phi}}_{{\beta }_{\sf DAC}} {\rm diag}\left(P{\bf f}_{0}{\bf f}_{0}^{\sf H}+P\sum_{i = 1}^{K} {\bf{f}}_i {\bf{f}}_i^{\sf H}\right) \\
    \label{eq:re covariance of quantization noise_2}
     = &\;  {\bf{\Phi}}_{{\alpha }_{\sf DAC}}{\bf{\Phi}}_{{\beta }_{\sf DAC}} {\rm diag}\left(P{\bf F}{\bf F}^{\sf H}\right).  
\end{align}
Then, the power constraint in \eqref{eq:problem_main_constraint} is reformulated as
\begin{align}
    \text{{tr}}\left( \mathbb{E}\left[{\bf x}_{\sf q}{\bf x}_{\sf q}^{\sf H}\right]\right) &={\text {tr}}\left( P{\bf \Phi}_{\alpha_{\sf DAC}}\sum_{i = 0}^{K} {\bf{f}}_i{\bf{f}}_i^{\sf H}{\bf \Phi}_{\alpha_{\sf DAC}}^{\sf H}+{\bf R_{{\bf q}_{\sf DAC}{\bf q}_{\sf DAC}}}\right) 
    \\
    &\mathop{=}^{(a)}{\text {tr}}\left( P{\bf \Phi}_{\alpha_{\sf DAC}} {\bf F}{\bf F}^{\sf H}\right)\\
    & \leq P,
\end{align}
where ${(a)}$ comes from \eqref{eq:re covariance of quantization noise_2} and ${\bf \Phi}_{\beta_{\sf DAC}} = {\bf I}_N - {\bf \Phi}_{\alpha_{\sf DAC}}$.
Finally, the constraint in \eqref{eq:problem_main_constraint} reduces to
\begin{align}
    {\text {tr}}\left({\bf \Phi}_{\alpha_{\sf DAC}} {\bf F}{\bf F}^{\sf H}\right) {\leq} 1.
\end{align}

Now, to address the challenges of the non-smoothness of the minimum operation and non-convexity of the spectral efficiency, we  approximate the minimum rate of $R_{{\sf c},k}$ and further reformulate the problem into a tractable form. 
Let us define a positive constant $\tau > 0$.
We utilize the LogSumExp technique to approximate the minimum function as  \cite{shen2010dual} 
\begin{align}
    \label{eq:logsumexp}
    \min_{i = 1,...,N}\{x_i\}  \approx - {\tau} \ln\left(\sum_{i = 1}^{N} \exp\left( -\frac{1}{\tau}{x_i}  \right)\right),
\end{align}
where the approximation becomes tight as ${\tau} \rightarrow 0$.
Applying \eqref{eq:logsumexp} to the rate of the common stream in \eqref{eq:ergodic common message}, we have
\begin{align}
    \label{eq:approximation_R}
    \min_{k \in \CMcal{K }} \{R_{{\sf c},k}\} \approx  -{\tau} \ln \left( \sum_{k = 1}^{K} \exp\left( -{\frac{1}{\tau} R_{{\sf c},k}}\right) \right).
\end{align}
We remark that once we obtain precoders using the approximation, we re-compute the achievable spectral efficiency   by using the minimum operation without approximation to determine the actual common rate.
Although the approximation in \eqref{eq:approximation_R} converts the objective function in \eqref{eq:problem_main} to a smooth function, the non-convexity still resides in the problem which cannot be avoided, and the quantization errors which are the functions of the precoding vectors still make the problem more complicated to solve.
To resolve the difficulty, we now re-express $\text{QE}_{k}$ in \eqref{eq:QEk} by reorganizing the DAC quantization error covariance-related term as
\begin{align}
    {\bf h}_k^{\sf H}{\bf R}_{{\bf q}_{\sf DAC}{\bf q}_{\sf DAC}}{\bf h}_k
    &= {\bf h}_k^{\sf H}{\bf{\Phi}}_{{\alpha }_{\sf DAC}}{\bf{\Phi}}_{{\beta }_{\sf DAC}} {\rm diag}\left(P\sum_{i = 0}^{K} {\bf{f}}_i {\bf{f}}_i^{\sf H}\right){\bf h}_k \\
    \label{eq:rewrite the DAC quantizationne_1}
    &= P\sum_{i=0}^{K}{\bf f}_i^{\sf H}{\bf{\Phi}}_{{\alpha }_{\sf DAC}}{\bf{\Phi}}_{{\beta }_{\sf DAC}}{\rm diag}\left({\bf h}_k{\bf h}_k^{\sf H}\right){\bf f}_i,
\end{align}
and the ADC quantization error covariance-related term as
\begin{align}
    \frac{r_{{\sf q_{{\sf ADC},k}}{\sf q_{{\sf ADC},k}}}}{{\alpha }_{{\sf ADC},k}{\beta }_{{\sf ADC},k}} &= {{\bf h}_k}^{\sf H} \mathbb{E}{[{\bf x}_{\sf q} {\bf x}_{\sf q}^{\sf H}]} {{\bf h}_k} + \sigma^2 \\
    &= {{\bf h}_k}^{\sf H} \left( P{\bf \Phi}_{\alpha_{\sf DAC}}\sum_{i = 0}^{K} {\bf{f}}_i{\bf{f}}_i^{\sf H}{\bf \Phi}_{\alpha_{\sf DAC}}^{\sf H}+{\bf R_{{\bf q}_{\sf DAC}{\bf q}_{\sf DAC}}}\right) {{\bf h}_k}+ \sigma^2 \\
    \label{eq:covariance of quantization noise at user k_ inserting rqkqk_2}
    &\stackrel{(a)}= P\sum_{i=0}^{K}{\bf f}_i^{\sf H}\left({\bf{\Phi}}_{{\alpha }_{\sf DAC}}^{\sf H}{\bf h}_k{\bf h}_k^{\sf H}{\bf{\Phi}}_{{\alpha }_{\sf DAC}}+{\bf{\Phi}}_{{\alpha }_{\sf DAC}}{\bf{\Phi}}_{{\beta }_{\sf DAC}}{\rm diag}\left({\bf h}_k{\bf h}_k^{\sf H}\right)\right){\bf f}_i+\sigma^2.
\end{align}
where $r_{{\sf q_{{\sf ADC},k}}{\sf q_{{\sf ADC},k}}}$ represents the $k$th diagonal entry of \eqref{eq:RqADC} and $(a)$ follows from  \eqref{eq:rewrite the DAC quantizationne_1}.
Using \eqref{eq:rewrite the DAC quantizationne_1} and \eqref{eq:covariance of quantization noise at user k_ inserting rqkqk_2}, the SINRs of the common stream of user $k$ is reorganized as 
\begin{align}
    \label{eq:sinr_c_reorg}
    \gamma_{{\sf c},k} = \frac{{\alpha _ {{\sf ADC},k}}|{\bf h}_k^{\sf H}{\bf{\Phi}}_{{\alpha }_{\sf DAC}}{\bf{f}}_{0}|^2} {{\sum_{i = 0}^{K} |{\bf h}_{k}^{\sf H}{\bf{\Phi}}_{{\alpha }_{\sf DAC}} {\bf{f}}_{i}|^2 - {\alpha }_{{\sf ADC},k}|{\bf h}_{k}^{\sf H}{\bf{\Phi}}_{{\alpha }_{\sf DAC}} {\bf{f}}_{0}|^2+\sum_{i=0}^{K}{\bf f}_i^{\sf H}{\bf{\Phi}}_{{\alpha }_{\sf DAC}}{\bf{\Phi}}_{{\beta }_{\sf DAC}}{\rm diag}\left({{\bf h}}_k{{\bf h}}_k^{\sf H}\right){\bf f}_i+\frac{{\sigma}^2}{P}}}.
\end{align}
Similarly, the SINR of the private stream of user $k$ becomes 
\begin{align}
    \label{eq:sinr_p_reorg}
   & \gamma_k = 
   \\
   \nonumber
    &\frac{{\alpha _ {{\sf ADC},k}}|{\bf h}_k^{\sf H}{\bf{\Phi}}_{{\alpha }_{\sf DAC}}{\bf{f}}_{k}|^2} {{\sum_{i = 0}^{K} |{\bf h}_{k}^{\sf H}{\bf{\Phi}}_{{\alpha }_{\sf DAC}} {\bf{f}}_{i}|^2 \!-\!\alpha_{{\sf ADC},k}(|{\bf h}_{k}^{\sf H}{\bf{\Phi}}_{{\alpha }_{\sf DAC}} {\bf{f}}_{k}|^2\!+\!|{\bf h}_{k}^{\sf H}{\bf{\Phi}}_{{\alpha }_{\sf DAC}} {\bf{f}}_{0}|^2) \!+\!\sum_{i=0}^{K}{\bf f}_i^{\sf H}{\bf{\Phi}}_{{\alpha }_{\sf DAC}}{\bf{\Phi}}_{{\beta }_{\sf DAC}}\!{\rm diag}\left({{\bf h}}_k{{\bf h}}_k^{\sf H}\right){\bf f}_i\!+\!\frac{{\sigma}^2}{P}}}.
\end{align}
\begin{remark}
  [Effect of ADC and DAC quantization errors to RSMA]
  \label{rm:effect}
  \normalfont 
   To compare the effect of quantization errors from ADCs and DACs  specifically to the use of the common stream, let us first assume that the quantization error comes only from ADCs. 
    Then \eqref{eq:sinr_c_reorg} becomes
    \begin{align}
        \label{eq:SINR_ADC}
        \gamma_{{\sf c},k}^{\sf ADC} = \frac{{\alpha _ {{\sf ADC},k}}|{\bf h}_k^{\sf H}{\bf{f}}_{0}|^2} {{\sum_{i = 1}^{K} |{\bf h}_{k}^{\sf H} {\bf{f}}_{i}|^2} + \underbrace{(1- {\alpha }_{{\sf ADC},k})|{\bf h}_{k}^{\sf H} {\bf{f}}_{0}|^2}_{\text{quantization error from  }s_{\sf c}}+\frac{{\sigma}^2}{P}},
    \end{align}
    Now, let us assume that the quantization error occurs only from DACs and they are homogeneous.
    Then \eqref{eq:sinr_c_reorg} becomes
    \begin{align}
        \gamma_{{\sf c},k}^{\sf DAC}&=\frac{\alpha_{{\sf DAC}}|{\bf h}_k^{\sf H}{\bf{f}}_{0}|^2} {{\alpha }_{\sf DAC}{\sum_{i = 1}^{K} |{\bf h}_{k}^{\sf H} {\bf{f}}_{i}|^2 +(1-\alpha_{\sf DAC})\sum_{i=0}^{K}{\bf f}_i^{\sf H}{\rm diag}\left({{\bf h}}_k{{\bf h}}_k^{\sf H}\right){\bf f}_i+\frac{{\sigma}^2}{P}}}
        \\
        \label{eq:SINR_DAC}
        &\approx\frac{\alpha_{{\sf DAC}}|{\bf h}_k^{\sf H}{\bf{f}}_{0}|^2} {{\alpha }_{\sf DAC}{\sum_{i = 1}^{K} |{\bf h}_{k}^{\sf H} {\bf{f}}_{i}|^2 +(1-\alpha_{\sf DAC})\sum_{i=1}^{K}{\bf f}_i^{\sf H}{\rm diag}\left({{\bf h}}_k{{\bf h}}_k^{\sf H}\right){\bf f}_i+  \underbrace{e(1-\alpha_{\sf DAC})|{\bf f}_0^{\sf H}{{\bf h}}_k|^2}_{\text{quantization error from  }s_{\sf c}} +\frac{{\sigma}^2}{P}}}
    \end{align}
    where $0<e < 1$ is a positive constant value less than one  by assuming $|{\bf f}_0^{\sf H}{{\bf h}}_k|^2 > {\bf f}_0^{\sf H}{\rm diag}({{\bf h}}_k{{\bf h}}_k^{\sf H}){\bf f}_0$;  since the precoding vector for the common stream $\bff_{0}$ should be constructively combined for the channels of the RSMA users, taking only the diagonal terms of $(\bh_k\bh_k^{\sf H})$ deteriorates the design goal and reduces beamforming gains.
    Then, when comparing \eqref{eq:SINR_ADC} and \eqref{eq:SINR_DAC}, the quantization error from $s_{\sf c}$ tends to be worse in \eqref{eq:SINR_ADC} with the same resolution, thereby preventing the common stream more from using higher transmit power.
   Since there is no difference between the common and private SINRs in terms of the quantization errors, the same occurs for the SINRs of the private streams,
  i.e., the SINR of the private stream  has the quantization error term associated with the common stream under ADC and DAC quantizations as
  $(1- {\alpha }_{{\sf ADC},k})|{\bf h}_{k}^{\sf H} {\bf{f}}_{0}|^2$ and $e(1- {\alpha }_{{\sf DAC}})|{\bf h}_{k}^{\sf H} {\bf{f}}_{0}|^2$, respectively. 
   This indicates that allocating more power to the common stream gives more penalty to the private streams when using low-resolution ADCs than low-resolution DACs.
    In other words, the DAC quantizes the signals before the common precoder is constructively combined with the user channels, whereas the ADC quantizes the signals after the precoder is constructively combined with the channels, thereby providing larger quantization error.
   Therefore, it is concluded that using low-resolution ADCs tends to  limit the use of the common stream  than using low-resolution DACs for the same resolution.
 This will be confirmed  in Section~\ref{subsec:DACADCeffect}. 
\end{remark}
\begin{remark}  [Effect of ADC and DAC quantization errors to SDMA]\normalfont
The behavior of the quantization error to the common stream  in Remark~\ref{rm:effect} is due to the multicast nature of the common stream; $\bff_0$ needs to be constructively combined with all RSMA user channels, and thus, the significance of ADC and DAC quantization error is not the same to the use of private streams as analyzed as follows:
 if we consider SDMA, we have the SINRs of user $k$ as
  \begin{align}
        \label{eq:SINR_sdma_ADC}
        \gamma_{k}^{\sf ADC} = \frac{{\alpha _ {{\sf ADC},k}}|{\bf h}_k^{\sf H}{\bf{f}}_{0}|^2} {{\sum_{i = 1}^{K} |{\bf h}_{k}^{\sf H} {\bf{f}}_{i}|^2} -\alpha_{{\sf ADC},k} |\bh_k^{\sf H}\bff_k|^2+\frac{{\sigma}^2}{P}}
    \end{align}
    and
     \begin{align}
        \label{eq:SINR_sdma_DAC}
        \gamma_{k}^{\sf DAC}=\frac{\alpha_{{\sf DAC}}|{\bf h}_k^{\sf H}{\bf{f}}_{0}|^2} {{\sum_{i=1}^{K}\left\{{\alpha }_{\sf DAC} |{\bf h}_{k}^{\sf H} {\bf{f}}_{i}|^2 +(1\!-\!\alpha_{\sf DAC}){\bf f}_i^{\sf H}{\rm diag}\left({{\bf h}}_k{{\bf h}}_k^{\sf H}\right){\bf f}_i\right\} -\alpha_{{\sf DAC}} |\bh_k^{\sf H}\bff_k|^2 +\frac{{\sigma}^2}{P}}}.
    \end{align}
We note that when $|{\bf f}_i^{\sf H}{{\bf h}}_k|^2 = {\bf f}_i^{\sf H}{\rm diag}({{\bf h}}_k{{\bf h}}_k^{\sf H}){\bf f}_i$, $\forall i$, \eqref{eq:SINR_sdma_ADC} and \eqref{eq:SINR_sdma_DAC} are equal for the same ADC and DAC resolutions.
For $i = k$, we tend to have $|{\bf f}_i^{\sf H}{{\bf h}}_k|^2 > {\bf f}_i^{\sf H}{\rm diag}({{\bf h}}_k{{\bf h}}_k^{\sf H}){\bf f}_i$ because $\bff_k$ should be designed to be constructively combined with $\bh_k$.
For $i\neq k$, however, we can assume $|{\bf f}_i^{\sf H}{{\bf h}}_k|^2 < {\bf f}_i^{\sf H}{\rm diag}({{\bf h}}_k{{\bf h}}_k^{\sf H}){\bf f}_i$ since ${\bf f}_i$ is normally designed to nullify the interfering channel $\bh_k$, but diagonalization of $\bh_k\bh_k^{\sf H}$ may break the goal.
Hence, the DAC quantization error from the private streams may increase or decrease the SINR compared to the SINR under ADC quantization error.
 Therefore, no such claim in Remark~\ref{rm:effect} can be made for the use of private streams.
 \end{remark}

Now, we define the weighted precoding vector of user $k$ as
\begin{align}
    {\bf w}_k = {\bf {\Phi}}_{ \alpha_{\sf DAC}}^{1/2}{\bf f}_k.
\end{align}
Let $\bW = [\bw_0, \bw_1, \dots, \bw_K]$.
We vectorize the weighted precoding matrix $\bW$  as ${\bar{\bf w}} = {{\rm vec}}\left(\bf W\right)$.
Here, we assume tr$\left({\bf W}{\bf W}^{\sf H} \right) = 1$ which indicates that the AP uses the maximum transmit power $P$, which is optimal in terms of maximizing the spectral efficiency.
Let ${\bf G}_k = ({\bf{\Phi}}_{{\alpha }_{\sf DAC}}^{1/2})^{\sf H}{\bf h}_k{\bf h}_k^{\sf H}{\bf{\Phi}}_{{\alpha }_{\sf DAC}}^{1/2}+{\bf{\Phi}}_{{\beta }_{\sf DAC}}{\rm diag}\left({\bf h}_k {\bf h}_k^{\sf H}\right)$. 
Then using the SINRs in \eqref{eq:sinr_c_reorg} and \eqref{eq:sinr_p_reorg} with the vectorized weighted precoder $\bar \bw$, we represent  $R_{{\sf c},k}$ in a Rayleigh quotient form as
\begin{align}
    \label{eq:rewrite_block}
    R_{{\sf c},k} = \log_2 \left( \frac{\bar {\bf{w}}^{\sf H} {\bf{A}}_{{\sf c},k} \bar{\bf{w}}}{\bar{\bf{w}}^{\sf H} {\bf{B}}_{{\sf c},k} \bar{\bf{w}} } \right),
\end{align}
where
\begin{align}
    &{\bf{A}}_{{\sf c},k} = {\rm blkdiag} \left({\bf G}_k, \cdots , {\bf G}_k \right) + {\bf{I}}_{N(K+1)} \frac{ \sigma^2}{P}, \\
    &{\bf{B}}_{{\sf c},k}\! =\! {\bf{A}}_{{\sf c},k}\! -\! {\rm blkdiag} \left(\!\alpha_{{\sf ADC},k}({{\bf{\Phi}}_{{\alpha }_{\sf DAC}}^{1/2}})^{\sf H}{\bf{h}}_k {\bf{h}}_k^{\sf H}{\bf{\Phi}}_{{\alpha }_{\sf DAC}}^{1/2},{{\bf{0}}_{N}},\!\cdots\!, {{\bf{0}}_N}\!\right).
\end{align}
Here, ${\bf{A}}_{{\sf c},k}$ and ${\bf{B}}_{{\sf c},k}$ are the diagonal matrices of size $N(K+1)\times N(K+1)$.
Similarly, we cast  $R_{k}$ into the Rayleigh quotient form as
\begin{align}
    \label{eq:rewrite_block_pri}
    R_{k} = \log_2 \left( \frac{\bar {\bf{w}}^{\sf H} {\bf{A}}_{k}\bar{\bf{w}}}{\bar{\bf{w}}^{\sf H} {\bf{B}}_{k} \bar{\bf{w}} } \right),
\end{align}
where
\begin{align}
    \nonumber
    {\bf{A}}_{k}\! =& {\rm blkdiag} \left({\bf G}_k\!-\!\alpha_{{\sf ADC},k}(\!{{\bf{\Phi}}_{{\alpha}_{\sf DAC}}^{1/2}})^{\sf H} {\bf{h}}_k {\bf{h}}_k^{\sf H}{{\bf{\Phi}}_{{\alpha }_{\sf DAC}}^{1/2}},{\bf G}_k,\! \cdots \!, {\bf G}_k\right)\!+ {\bf{I}}_{N(K+1)} \frac{ \sigma^2}{P}, \\
    {\bf{B}}_{k}\! =& {\bf{A}}_{k}\!\! -\! {\rm blkdiag} \left(\!{{\bf{0}}_N}\!,\! \cdots\! ,    \underbrace{\alpha_{{\sf ADC},k}(\!{{\bf{\Phi}}_{{\alpha }_{\sf DAC}}^{1/2}}\!)^{\sf H} {\bf{h}}_k {\bf{h}}_k^{\sf H}{{\bf{\Phi}}_{{\alpha }_{\sf DAC}}^{1/2}}}_{{\text{the} \;(k+1){\text{th block}}}},\!{{\bf{0}}_N},\! \cdots\! ,\! {{\bf{0}}_N}\!\! \right).
\end{align}
Finally, based on \eqref{eq:approximation_R}, \eqref{eq:rewrite_block}, and \eqref{eq:rewrite_block_pri}, the optimization problem in \eqref{eq:problem_main} is reformulated as
\begin{align}
    \label{eq:problem_new_Blk_q}
    \mathop{{\text{maximize}}}_{\bar {\bf{w}}}& \;\;
    \ln \left(\sum_{k = 1}^{K} \left( \frac{\bar {\bf{w}}^{\sf H} {\bf{A}}_{{\sf c},k} \bar {\bf{w}}}{\bar {\bf{w}}^{\sf H} {\bf{B}}_{{\sf c},k}  \bar {\bf{w}} }  \right)^{-\frac{1}{\tau\ln2}} \right)^{-\tau}\! + \!
    \frac{1}{\ln2}\sum_{k = 1}^{K}\ln \left( \frac{\bar {\bf{w}}^{\sf H} {\bf{A}}_k \bar {\bf{w}}}{\bar {\bf{w}}^{\sf H} {\bf{B}}_k \bar {\bf{w}}}\right) \\
    \label{eq:constraint_new_new_Blk}
    {\text{subject to}} & \;\; \|{\bar{{\bf w}}}\| = 1.
\end{align}
Recall that the equality constraint in \eqref{eq:constraint_new_new_Blk} comes from using the maximum transmit power. 
We remark that the problem in \eqref{eq:problem_new_Blk_q} is invariant to $\bar \bw$ upto its scaling. 
Accordingly, we can effectively ignore the constraint in \eqref{eq:constraint_new_new_Blk}.

\subsection{First-Order Optimality Condition}
In this subsection, we derive the first-order optimality condition of \eqref{eq:problem_new_Blk_q} with respect to $\bar {\bf{w}}$.
\begin{lemma}
    \label{lem:main}
    The first-order optimality condition of the optimization problem \eqref{eq:problem_new_Blk_q} is satisfied if the following holds: 
\begin{align}
    \label{eq:lem_kkt_stack}
    {\bf{B}}_{\sf KKT}^{-1} (\bar {\bf{w}}){\bf{A}}_{\sf KKT}(\bar {\bf{w}}) \bar {\bf{w}} = \lambda(\bar {\bf{w}}) \bar {\bf{w}},
\end{align} 
where 
\begin{align}
    \label{eq:lem_A_kkt_Q}
    &{\bf{A}}_{\sf KKT}(\bar {\bf{w}}) =  \lambda_{\sf num}(\bar {\bf{w}}) \cdot \sum_{k = 1}^{K}  \left[ \frac{\exp\left( -\frac{1}{\tau}  \log_2\left(\frac{\bar {\bf{w}}^{\sf H} {\bf{A}}_{{\sf c},k} \bar {\bf{w}}}{\bar {\bf{w}}^{\sf H} {\bf{B}}_{{\sf c},k} \bar {\bf{w}} } \right) \right)}{\sum_{\ell =1}^{K} \exp\left(-\frac{1}{\tau} \log_2\left(\frac{\bar {\bf{w}}^{\sf H} {\bf{A}}_{{\sf c},\ell} \bar {\bf{w}}}{\bar {\bf{w}}^{\sf H} {\bf{B}}_{{\sf c},\ell} \bar {\bf{w}} } \right)\right)} \frac{{\bf{A}}_{{\sf c},k}}{\bar {\bf{w}}^{\sf H} {\bf{A}}_{{\sf c},k} \bar {\bf{w}}} + \frac{{\bf{A}}_k}{\bar {\bf{w}}^{\sf H} {\bf{A}}_k \bar {\bf{w}}} \right]
    \\
    \label{eq:lem_B_kkt_Q}
    &{\bf{B}}_{\sf KKT}(\bar {\bf{w}}) = \lambda_{\sf den} (\bar {\bf{w}}) \cdot \sum_{k = 1}^{K}  \left[ \frac{\exp\left( -\frac{1}{\tau}  \log_2\left(\frac{\bar {\bf{w}}^{\sf H} {\bf{A}}_{{\sf c},k} \bar {\bf{w}}}{\bar {\bf{w}}^{\sf H} {\bf{B}}_{{\sf c},k} \bar {\bf{w}} } \right) \right)}{\sum_{\ell = 1}^{K} \exp\left(-\frac{1}{\tau} \log_2\left(\frac{\bar {\bf{w}}^{\sf H} {\bf{A}}_{{\sf c},\ell} \bar {\bf{w}}}{\bar {\bf{w}}^{\sf H} {\bf{B}}_{{\sf c},\ell} \bar {\bf{w}} } \right)\right)} \frac{{\bf{B}}_{{\sf c},k}}{\bar {\bf{w}}^{\sf H} {\bf{B}}_{{\sf c},k} \bar {\bf{w}}} + \frac{{\bf{B}}_k}{\bar {\bf{w}}^{\sf H} {\bf{B}}_k \bar {\bf{w}}} \right], 
\end{align}
with
\begin{align}
    \label{eq:lem_lambda}
    \lambda(\bar {\bf{w}}) &= \left\{\frac{1}{K\ln2 }\sum_{k = 1}^{K} \left(\frac{\bar {\bf{w}}^{\sf H} {\bf{A}}_{{\sf c},k} \bar {\bf{w}}}{\bar {\bf{w}}^{\sf H} {\bf{B}}_{{\sf c},k} {\bar{\bf{w}}}} \right)^{ -\frac{1}{\tau}} \right\}^{-\frac{\tau}{\ln 2}} \ \prod_{k = 1}^{K} \left(\frac{\bar {\bf{w}}^{\sf H} {\bf{A}}_k \bar {\bf{w}}}{\bar {\bf{w}}^{\sf H} {\bf{B}}_k \bar {\bf{w}}} \right),
    \\
    \lambda_{\sf num} &= \prod_{k = 1}^{K}\left({\bar {\bf{w}}^{\sf H} {\bf{A}}_k \bar {\bf{w}}}\right), \ 
    \lambda_{\sf den} = \left\{\frac{1}{K\ln2 }\sum_{k = 1}^{K} \left(\frac{\bar {\bf{w}}^{\sf H} {\bf{A}}_{{\sf c},k} \bar {\bf{w}}}{\bar {\bf{w}}^{\sf H} {\bf{B}}_{{\sf c},k} {\bar{\bf{w}}}} \right)^{ -\frac{1}{\tau}} \right\}^{\frac{\tau}{\ln 2}} \prod_{k = 1}^{K}\left({\bar {\bf{w}}^{\sf H} {\bf{B}}_k \bar {\bf{w}}}\right).
\end{align}
\end{lemma}
\begin{proof}
See Appendix \ref{proof:lem1}.
\end{proof}
Lemma \ref{lem:main} states that the first-order optimality condition can be presented as a nonlinear eigenvalue problem for ${\bf{B}}^{-1}_{\sf KKT} (\bar {\bf{w}}){\bf{A}}_{\sf KKT}(\bar {\bf{w}})$.
We note that the objective function of the problem \eqref{eq:problem_new_Blk_q} is equal to $\ln{\lambda(\bar {\bf{w}})}$.
This indicates that finding the leading eigenvector of \eqref{eq:lem_kkt_stack} is equivalent to finding the best local optimal solution because any eigenvector of \eqref{eq:lem_kkt_stack} is indeed one of the stationary points of the problem \eqref{eq:problem_new_Blk_q}, which is stated as follows:
\begin{proposition}
    \label{prof:main} Denoting the leading eigenvector for the problem \eqref{eq:lem_kkt_stack} to be $\bar {\bf{w}}^{\star }$ and its corresponding eigenvalue to be $\lambda^\star$, i.e., ${\bf{B}}_{\sf KKT}^{-1}(\bar {\bf{w}}^{\star}){\bf{A}}_{\sf KKT}(\bar {\bf{w}}^{\star}) \bar {\bf{w}}^{\star} = \lambda^{\star} \bar {\bf{w}}^{\star}$, the eigenvector $\bar {\bf{w}}^{\star}$ is the stationary point that achieves the best local optimal solution of the problem \eqref{eq:problem_new_Blk_q}.
\end{proposition}
Based on the observation, we propose a computationally efficient RSMA precoding algorithm that finds the best local optimal point to maximize the sum spectral efficiency.

\begin{algorithm} [t]
\caption{Quantized Generalized Power Iteration for Rate-Splitting (Q-GPI-RS)} \label{alg:main_QRS} 
{\bf{initialize}}: $\bar {\bf{w}}_{0}$\\
Set the iteration count $t = 0$.\\
\While {$\left\|\bar {\bf{w}}_{t+1} - \bar {\bf{w}}_{t} \right\| > \epsilon$ $\it{\&}$ $t \leq t_{\rm max}$}{
Build matrix $ {\bf{A}}_{\sf KKT} (\bar {\bf{w}}_{t})$ in \eqref{eq:lem_A_kkt_Q}\\
Build matrix $ {\bf{B}}_{\sf KKT} (\bar {\bf{w}}_{t})$ in \eqref{eq:lem_B_kkt_Q} \\
Compute $\bar {\bf{w}}_{t+1} = $ ${{\bf{B}}^{-1}_{\sf KKT} (\bar {\bf{w}}_{t}) {\bf{A}}_{\sf KKT} (\bar {\bf{w}}_{t}) \bar {\bf{w}}_{t}}$. \\
Normalize $\bar {\bf{w}}_{t+1} \leftarrow \frac{{\bf{B}}^{-1}_{\sf KKT} (\bar {\bf{w}}_{t}) {\bf{A}}_{\sf KKT} (\bar {\bf{w}}_{t}) \bar {\bf{w}}_{t}} {\| {\bf{B}}^{-1}_{\sf KKT} (\bar {\bf{w}}_{t}) {\bf{A}}_{\sf KKT} (\bar {\bf{w}}_{t}) \bar {\bf{w}}_{t} \|}$.\\
 $t \leftarrow t+1$.}
\Return{\ }{$\bar{\bf w}_t$}.
\end{algorithm}
\subsection{Quantized Generalized Power Iteration for  Rate-Splitting}
To identify the leading eigenvector of \eqref{eq:lem_kkt_stack}, we adopt the generalized power iteration-based method  \cite{choi:twc:20}.
Let $\bar \bw_t$ be the vectorized weighted precoder at the $t$-th iteration.
Then at the $t$-th iteration, the matrices ${\bf{A}}_{\sf KKT} (\bar {\bf{w}}_{t})$ and $ {\bf{B}}_{\sf KKT} (\bar {\bf{w}}_{t})$ are constructed using \eqref{eq:lem_A_kkt_Q} and \eqref{eq:lem_B_kkt_Q}, and then $\bar {\bf{w}}_{t+1}$ is updated as 
\begin{align}
    \bar {\bf{w}}_{t+1} = \frac{{\bf{B}}_{\sf KKT}^{-1} (\bar {\bf{w}}_{t}) {\bf{A}}_{\sf KKT} (\bar {\bf{w}}_{t}) \bar {\bf{w}}_{t}}{\| {\bf{B}}_{\sf KKT}^{-1} (\bar {\bf{w}}_{t}) {\bf{A}}_{\sf KKT} (\bar {\bf{w}}_{t}) \bar {\bf{w}}_{t} \|}.
\end{align}
The iteration is repeated until the algorithm satisfies the convergence criterion, ie., $\left\|\bar {\bf{w}}_{t+1} - \bar {\bf{w}}_{t} \right\| < \epsilon$ in which $\epsilon > 0$ denotes the threshold or reaches a maximum iteration count $t_{\rm max}$ which depends on a system requirement.
Algorithm~\ref{alg:main_QRS} summarizes the steps.
\begin{remark}[Algorithm Complexity]
    \label{remark:Algorithm complexity comparison}\normalfont
    The complexity of Q-GPI-RS depends on the calculation of ${\bf{B}}_{\sf KKT}^{-1} (\bar {\bf{w}})$.
    Since 
    the  matrix size is $N(K+1)\times N(K+1)$, the inversion of ${\bf{B}}_{\sf KKT}(\bar {\bf{w}})$ typically requires the complexity order of $\mathcal{O}((K+1)^3N^3)$.
    Noticing the block diagonal structure of ${\bf{B}}_{\sf KKT}(\bar {\bf{w}})$, only the complexity order of $\mathcal{O}((K+1)N^3)$ is required by  computing the inversion of each $N\times N$ sub-matrix separately 
    since there are $(K+1)$ sub-matrices.
    On the other hand, the convex relaxation methods which are based on QCQP \cite{joudeh2016sum,joude:twc:17} and CCCP \cite{li:jsac:20} have the complexity order of $\mathcal{O}(K^{3.5}N^{3.5})$ and $\mathcal{O}(N^6K^{0.5}2^{3.5K})$, respectively.
    In this regard, Q-GPI-RS is more efficient in computation compared to the other state-of-the-art methods.
\end{remark}
Prior to the numerical evaluation of the proposed method, we introduce an optimization approach based on direction extension of an existing state-of-the-art precoding approach  for performance comparison in the next section.

\section{Extension of Existing Approach: WMMSE Approach}    
\label{sec:q-wmmse-ao}
 Here, we briefly present the extension of the WMMSE-based alternating optimization (WMMSE-AO) approach \cite{joudeh2016sum} for precoding optimization in the considered system.
To solve the non-convex problem of maximizing the sum spectral efficiency with respect to the precoder, \cite{joudeh2016sum} developed the WMMSE-AO method without considering the quantization error.
Adopting the similar principle of the WMMSE-based algorithms in  \cite{joudeh2016sum} and \cite{dizdar2021rate}, we also solve our optimization problem.
To this end, we first denote that ${\hat s}_{{\sf c},k}$ and ${\hat s}_{k}$ are estimated values for ${s}_{{\sf c},k}$ and ${s}_{k}$.
Defining the scalar equalizers of the common and private streams for user $k$ as $g_{{\sf c},k}$ and $g_k$, respectively,  we compute the MSEs of the common stream $\epsilon_{{\sf c},k}$ and private stream $\epsilon_k$ received at user $k$ as
\begin{align}
    \label{eq:MSE of the common messages}
    \epsilon_{{\sf c},k} \!=\! \mathbb{E}[|{\hat s}_{{\sf c},k}\!-\!{s}_{\sf c}|^2] = \mathbb{E}[|g_{{\sf c},k}y_{{\sf q},k}\!-\!{s}_{\sf c}|^2] 
  =|g_{{\sf c},k}|^2T_{{\sf c},k} \!-\! 2{\sf Re}\left\{\sqrt{P}g_{{\sf c},k}\alpha_{{\sf ADC},k}{{{\bf h}_k^{\sf H}}}{\bf{\Phi}}_{{\alpha }_{\sf DAC}}{\bf f}_{0}\right\}\!+\!1
\end{align}
and
\begin{align}
    \label{eq:MSE of the private messages}
    \epsilon_k = \mathbb{E}[|{\hat s}_{k}-{s}_{k}|^2] = \mathbb{E}[|g_{k}y_{{\sf q},k}-{s}_{k}|^2] 
    = |g_k|^2T_k - 2{\sf Re}\left\{\sqrt{P}g_k\alpha_{{\sf ADC},k}{{{\bf h}_k^{\sf H}}}{\bf{\Phi}}_{{\alpha }_{\sf DAC}}{\bf f}_{k}\right\} + 1.
\end{align}
Here, $T_{{\sf c},k}$ and $T_k$ are defined as
\begin{align}
    \label{eq: T_c}
    &T_{{\sf c},k} =P\alpha_{{\sf ADC},k}^2\sum_{i=0}^{K}|{{{\bf h}_k^{\sf H}}}{\bf{\Phi}}_{{\alpha }_{\sf DAC}}{\bf f}_{i}|^2+{\alpha }_{{\sf ADC},k}^2{\bf h}_k^{\sf H}{\bf R}_{{\bf q}_{\sf DAC}{\bf q}_{\sf DAC}}{\bf h}_k + r_{{\sf q}_{{\sf ADC},k}{\sf q}_{{\sf ADC},k}}+\alpha_{{\sf ADC},k}^2\sigma^2
\end{align}
and
\begin{align}
    &T_k  = P\alpha_{{\sf ADC},k}^2\sum_{i=1}^{K}|{{{\bf h}_k^{\sf H}}}{\bf{\Phi}}_{{\alpha }_{\sf DAC}}{\bf f}_{i}|^2\!+\!{\alpha }_{{\sf ADC},k}^2{\bf h}_k^{\sf H}{\bf R}_{{\bf q}_{\sf DAC}{\bf q}_{\sf DAC}}{\bf h}_k \!+\! r_{{\sf q}_{{\sf ADC},k}{\sf q}_{{\sf ADC},k}}\!+\!\alpha_{{\sf ADC},k}^2\sigma^2 .
\end{align}
Then, we achieve the minimum MSEs when $g^{\sf MMSE}_{{\sf c},k} = \sqrt{P}\alpha_{{\sf ADC},k}{\bf f}_{0}^{\sf H}{\bf{\Phi}}_{{\alpha }_{\sf DAC}}^{\sf H}{{\bf h}_k}T_{{\sf c},k}^{-1}$ and $g^{\sf MMSE}_k = \sqrt{P}\alpha_{{\sf ADC},k}{\bf f}_{k}^{\sf H}{\bf{\Phi}}_{{\alpha }_{\sf DAC}}^{\sf H}{{\bf h}_k}T_k^{-1}$.
Based on \eqref{eq: T_c}, the MMSE of the common stream is 
\begin{align}
    \epsilon_{{\sf c},k}^{\sf MMSE} &= T_{{\sf c},k}^{-1}(T_{{\sf c},k}-P\alpha_{{\sf ADC},k}^2|{{{\bf h}_k^{\sf H}}}{\bf{\Phi}}_{{\alpha }_{\sf DAC}}{\bf f}_{0}|^2)
\end{align}
and the MMSE of the private stream for user $k$ is 
\begin{align}
    \epsilon_k^{\sf MMSE} &= T_k^{-1}(T_k-P\alpha_{{\sf ADC},k}^2|{{{\bf h}_k^{\sf H}}}{\bf{\Phi}}_{{\alpha }_{\sf DAC}}{\bf f}_{k}|^2).
\end{align}
Based on \eqref{eq:MSE of the common messages}, the augmented WMSE of the common stream is  defined by
\begin{align}
    \xi_{{\sf c},k} =&\; u_{{\sf c},k}\epsilon_{{\sf c},k} - \log_2(u_{{\sf c},k})
    \\
    \nonumber
    =& \;P\!\sum_{i=0}^{K}\!{\bf f}_{i}^{\sf H}\left(\alpha_{{\sf ADC},k}^2u_{{\sf c},k}|g_{{\sf c},k}|^2{\bf{\Phi}}_{{\alpha }_{\sf DAC}}^{\sf H}{\bf h}_k{\bf h}_k^{\sf H}{\bf{\Phi}}_{{\alpha }_{\sf DAC}}\!\right){\bf f}_{i}-2{\sf Re}\left\{\sqrt{P}u_{{\sf c},k}g_{{\sf c},k}\alpha_{{\sf ADC},k}{{{\bf h}_k^{\sf H}}}{\bf{\Phi}}_{{\alpha }_{\sf DAC}}{\bf f}_{0}\right\}
    \\
    \nonumber
    &+{\alpha }_{{\sf ADC},k}^2u_{{\sf c},k}|g_{{\sf c},k}|^2{\bf h}_k^{\sf H}{\bf R}_{{\bf q}_{\sf DAC}{\bf q}_{\sf DAC}}{\bf h}_k + u_{{\sf c},k}|g_{{\sf c},k}|^2r_{{\sf q}_{{\sf ADC},k}{\sf q}_{{\sf ADC},k}}+\alpha_{{\sf ADC},k}^2u_{{\sf c},k}|g_{{\sf c},k}|^2\sigma^2
    \\
    &+u_{{\sf c},k}-\log_2(u_{{\sf c},k}).
\end{align}
Similarly, using \eqref{eq:MSE of the private messages}, the augmented WMSE of the private stream for user $k$ follows as 
\begin{align}
    \xi_k =&\; u_k\epsilon_k - \log_2(u_k)
        \\
    =&\;P\!\sum_{i=1}^{K}{\bf f}_{i}^{\sf H}\left(\alpha_{{\sf ADC},k}^2u_{k}|g_k|^2{\bf{\Phi}}_{{\alpha }_{\sf DAC}}^{\sf H}{\bf h}_k{\bf h}_k^{\sf H}{\bf{\Phi}}_{{\alpha }_{\sf DAC}}\right){\bf f}_{i}\ - 2{\sf Re}\left\{\sqrt{P}u_kg_k\alpha_{{\sf ADC},k}{{{\bf h}_k^{\sf H}}}{\bf{\Phi}}_{{\alpha }_{\sf DAC}}{\bf f}_{k}\right\}
    \\
    \nonumber
    &+{\alpha }_{{\sf ADC},k}^2u_k|g_k|^2{\bf h}_k^{\sf H}{\bf R}_{{\bf q}_{\sf DAC}{\bf q}_{\sf DAC}}{\bf h}_k + u_k|g_k|^2r_{{\sf q}_{{\sf ADC},k}{\sf q}_{{\sf ADC},k}}+\alpha_{{\sf ADC},k}^2u_k|g_k|^2\sigma^2+u_k-\log_2(u_k).
\end{align}
We obtain the optimal weights to achieve the minimum of $\xi_{{\sf c},k}$ and $\xi_k$ as $u_{{\sf c},k} = 1/\epsilon_{{\sf c},k}^{\sf MMSE}$ and $u_k = 1/\epsilon_k^{\sf MMSE}$.
Consequently, for given equalizers $\xi_{{\sf c},k}$, $\xi_k$, and weights $u_{{\sf c},k}$, $u_k$,  the sum spectral efficiency maximization problem is solved by the following WMSE minimization problem:
\begin{align}
    \label{eq:SE maximization for WMSE minimization}
    \mathop{{\text{minimize}}}_{{\bf{f}}_{0}, {\bf{f}}_1, \cdots,{\bf{f}}_K,\xi_{\sf c}} \;\; \xi_{\sf c}+{\sum_{k=1}^{K}}\xi_k\quad  \mathop{{\text{subject to}}}\;\;{\text {tr}}\left( {\bf \Phi}_{\alpha_{\sf DAC}} {\bf F}{\bf F}^{\sf H}\right) \leq 1,
    \;\xi_{{\sf c},k} \leq \xi_{\sf c}, \forall k \in \CMcal{K}.
\end{align}
We remark that the problem \eqref{eq:SE maximization for WMSE minimization} is the QCQP which can be solved by  CVX.
Accordingly, we compute the equalizers,  weights, and precoders in the alternating manner as follows:
\begin{enumerate}
 \item {\it{Update of equalizers and weights:}} we update the equalizers and weights  by computing $g^{\sf MMSE}_{{\sf c},k}$, $g^{\sf MMSE}_{k}$, $u_{{\sf c},k} = 1/\epsilon_{{\sf c},k}^{\sf MMSE}$ and $u_k = 1/\epsilon_k^{\sf MMSE}$
for given precoding vectors.
\item {\it{Update of precoders and $\xi_{\sf c}$}}: then, the precoders and $\xi_{\sf c}$ can be derived by solving \eqref{eq:SE maximization for WMSE minimization} via CVX for given equalizers and weights.
    
\item {\it{Repeat  steps 1 and 2:}} the steps 1 and 2 are repeated until convergence.
\end{enumerate}


\section{Numerical Results}
In this section, we compare the sum spectral efficiency of the proposed Q-GPI-RS and existing baseline methods.
The channel vector ${{\bf{h}}_{k}}$ is computed based on its spatial covariance matrix ${\bf R}_k=\mathbb{E}{\left[{\bf{h}}_k{\bf{h}}_k^{\sf H}\right]}$.
We employ the one-ring channel model to generate ${\bf{R}}_{k}$ \cite{adhi:tit:13};
The channel covariance matrix ${\bf{R}}_{k}$ at the $n$-th antenna and $m$-th antenna is defined as $\left[{\bf{R}}_{k}\right]_{n,m} = \frac{1}{2\vartriangle_k}\int_{\theta_k-\vartriangle_k}^{\theta_k+\vartriangle_k} e^{-j\frac{2\pi}{\psi}\Psi\left(x\right)\left({\bf r}_n - {\bf r}_m\right)}\, dx,$
where   $\psi$ is a signal wavelength, $\triangle_k$ is the angular spread of user $k$, $\theta_k$ is angle-of-departure (AoD) of user $k$, $\Psi\left(x\right)=\left[\cos\left(x\right),\sin{\left(x\right)}\right]$, and ${\bf{r}}_n$ is the position vector of $n$-th antenna.
Using the Karhunen-Loeve model, the channel vector ${\bf{h}}_k$ is decomposed as ${\bf{h}}_k = {\bf{U}}_k {\pmb{\Lambda}}_k^{\frac{1}{2}}{\bf{g}}_k$,
where ${\bf{U}}_k \in \mathbb{C}^{N\times r_k}$ is eigenvectors of ${\bf{R}}_k$, ${\pmb{\Lambda}}_k \in \mathbb{C}^{r_k \times r_k}$ is a diagonal matrix of eigenvalues of ${\bf{R}}_k$, each entry of ${\bf{g}}_k$ follows 
${\mathcal {CN} (0,1)}$,
${\bf{g}}_k$ $\in \mathbb{C}^{r_k}$ is identically distributed, and ${r_k}$ is the rank of ${\bf{R}}_k$.
By using a block fading channel model, ${\bf{g}}_k$ is considered to be constant within one transmission block.
The value of $\theta_k$ is varied according to the user's location.
The AoD $\theta_k$ follows IID uniform distribution between $0$ and $\pi$ $\forall k$ when we consider that users are randomly located.
In the case of correlated users, the differences of $\theta_k$ for all users are randomly distributed within $\pi/6$.
We set the simulation setting as $\vartriangle_k=\pi/6$, $\epsilon=0.01$, and  $\sigma^2 = 1$.

For initialization, we set the precoding vector $\bff_k = \bh_k$, i.e., maximum ratio transmission (MRT).
For the common stream, in particular, the average of channel vectors is adopted, i.e., $\bff_0 = \frac{1}{K}\bH \cdot {\bf 1}_{K\times1}$.
Then the stacked precoding vector $\bar \bw_0$ is initialized based on such $\bF$.
Let us define an effective channel vector as ${\bf h}^{\rm eff}_{k} = {\bf{\Phi}}_{{\alpha }_{\sf DAC}}^{\sf H}{{\bf h}_{k}}{\alpha }_{{\sf ADC},k}^{\sf H}$.
Then, the comparing baseline methods are:
\begin{itemize}
    \item {\bf Q-WMMSE-AO} (RSMA): The algorithm is introduced in Section~\ref{sec:q-wmmse-ao}. 
    Q-WMMSE-AO approach is aligned with the ADMM-based algorithm in \cite{dizdar2021rate}  when the radar pattern design is omitted and ADC quantization error is further incorporated.
    \item {\bf{WMMSE-AO}} (RSMA): 
   The WMMSE alternating optimization (AO) approach is the state-of-the-art method for RSMA precoding in \cite{joudeh2016sum}.
    \item {\bf{WMMSE}} (SDMA): The WMMSE-based precoding method is one of the most well-known precoding method proposed in \cite{chris:twc:08}.
    
    \item {\bf{Q-GPI-SEM}} (SDMA): 
The quantization-aware GPI-based spectral efficiency maximization (Q-GPI-SEM) is proposed 
in the low-resolution ADC/DAC regime for SDMA \cite{choi2021energyIOTJ}.

 \item {\bf{Q-MRT}}/{\bf{Q-ZF}}/{\bf{Q-RZF}} (SDMA): The conventional linear precoders based on the effective channel ${\bf h}^{\rm eff}_{k}$ such as quantization-aware MRT, ZF, and RZF are evaluated for the SDMA.  
\end{itemize}

\begin{figure}[!t]\centering
	\subfigure{\resizebox{0.42\columnwidth}{!}{\includegraphics{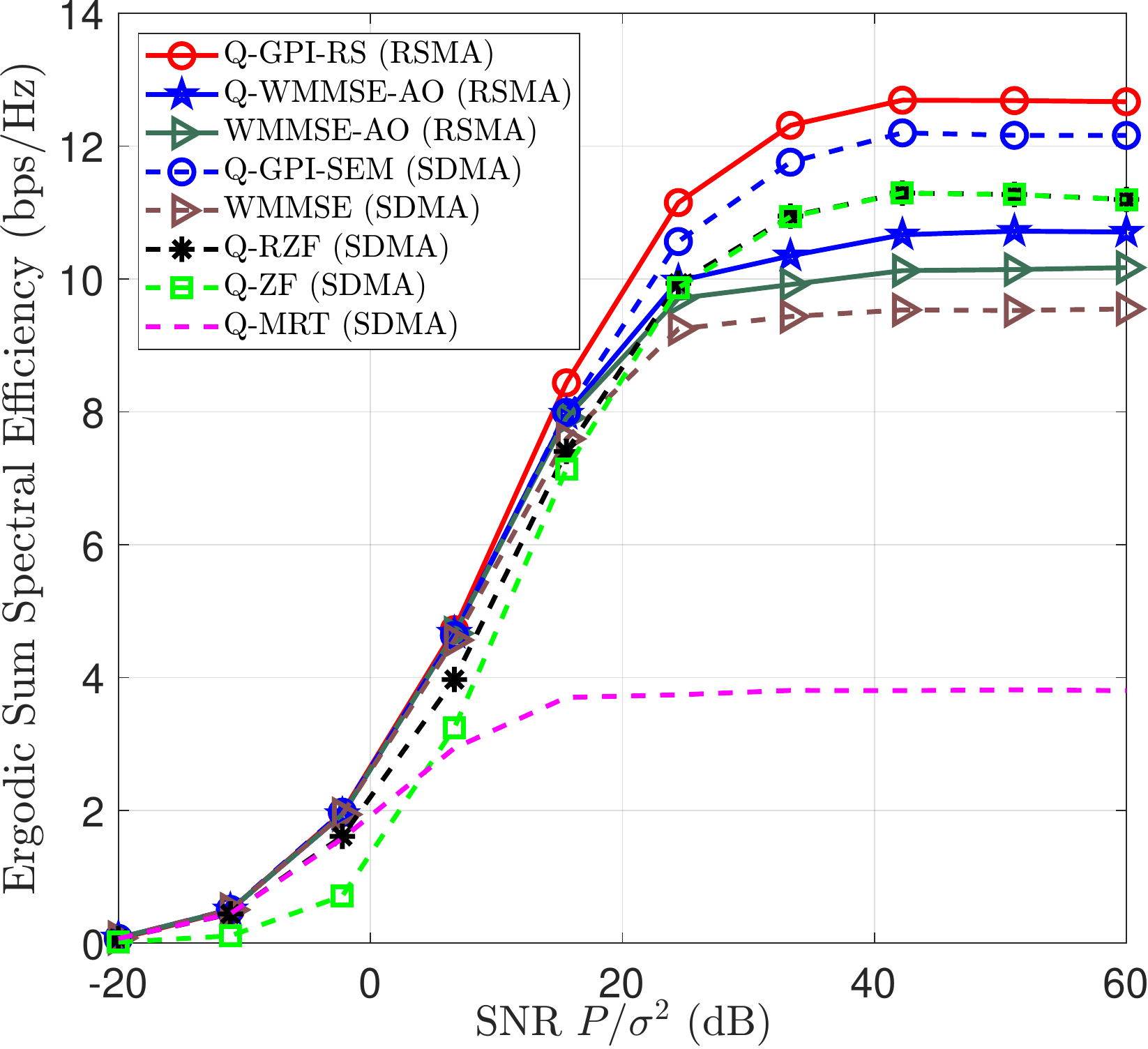}}}
	 	\vspace{-0.5em}
	\caption{The  sum spectral efficiency versus SNR for $N = 4$ AP antennas, $ K = 2$ users, $b_{{\sf DAC},n} = 4, \forall n$, and $b_{{\sf ADC},k} = 6, \forall k$.} 
 	\label{fig:N4K2RanCHSumSEDAC4ADC6}
 	\vspace{-1em}
\end{figure}
We define ${\rm SNR} = P/\sigma^2$.
Regarding the approximation parameter $\tau$, we first find the value  from the experimental results to the system configuration.
Then we fix $\tau$ with the obtained value for each evaluation point without online tuning of $\tau$.
The spectral efficiency in simulation results is computed based on \eqref{eq:ergodic common message} and \eqref{eq:ergodic private message} without the approximation of the minimum function.

\subsection{Homogeneous Quantization Resolution}

In Fig.~\ref{fig:N4K2RanCHSumSEDAC4ADC6}, we consider randomly located users for $N = 4$ and $K = 2$ with homogeneous quantization bits where $b_{{\sf DAC},n} = 4,\ \forall n$ and $b_{{\sf ADC},k} = 6,\ \forall k$.
As shown in Fig.~\ref{fig:N4K2RanCHSumSEDAC4ADC6}, Q-GPI-RS achieves the highest spectral efficiency.
Since  Q-WMMSE-AO cannot guarantee the best local optimal point, Q-WMMSE-AO  reveals the performance limitation compared to Q-GPI-RS even though the quantization effect is taken into consideration.
In addition, WMMSE-AO approach shows further performance degradation because WMMSE-AO method is designed without considering the quantization error.
As the SNR increases, Q-RZF and Q-ZF outperform the WMMSE-based methods.
Recalling that Q-RZF and Q-ZF are the SDMA precoders, this performance inversion indicates the importance of the accurate optimization of the RSMA precoder.
Accordingly, with properly designed RSMA precoding, the proposed Q-GPI-RS method shows the gain of RSMA compared to Q-GPI-SEM which is based on classical SDMA.
Such an RSMA gain is also observed by  comparing  WMMSE-AO with WMMSE under coarse quantization.
Therefore, the proposed Q-GPI-RS achieves the highest performance, and  RSMA  provides a noticeable gain over SDMA.
\begin{figure}[!t]\centering
	\subfigure{\resizebox{0.42\columnwidth}{!}{\includegraphics{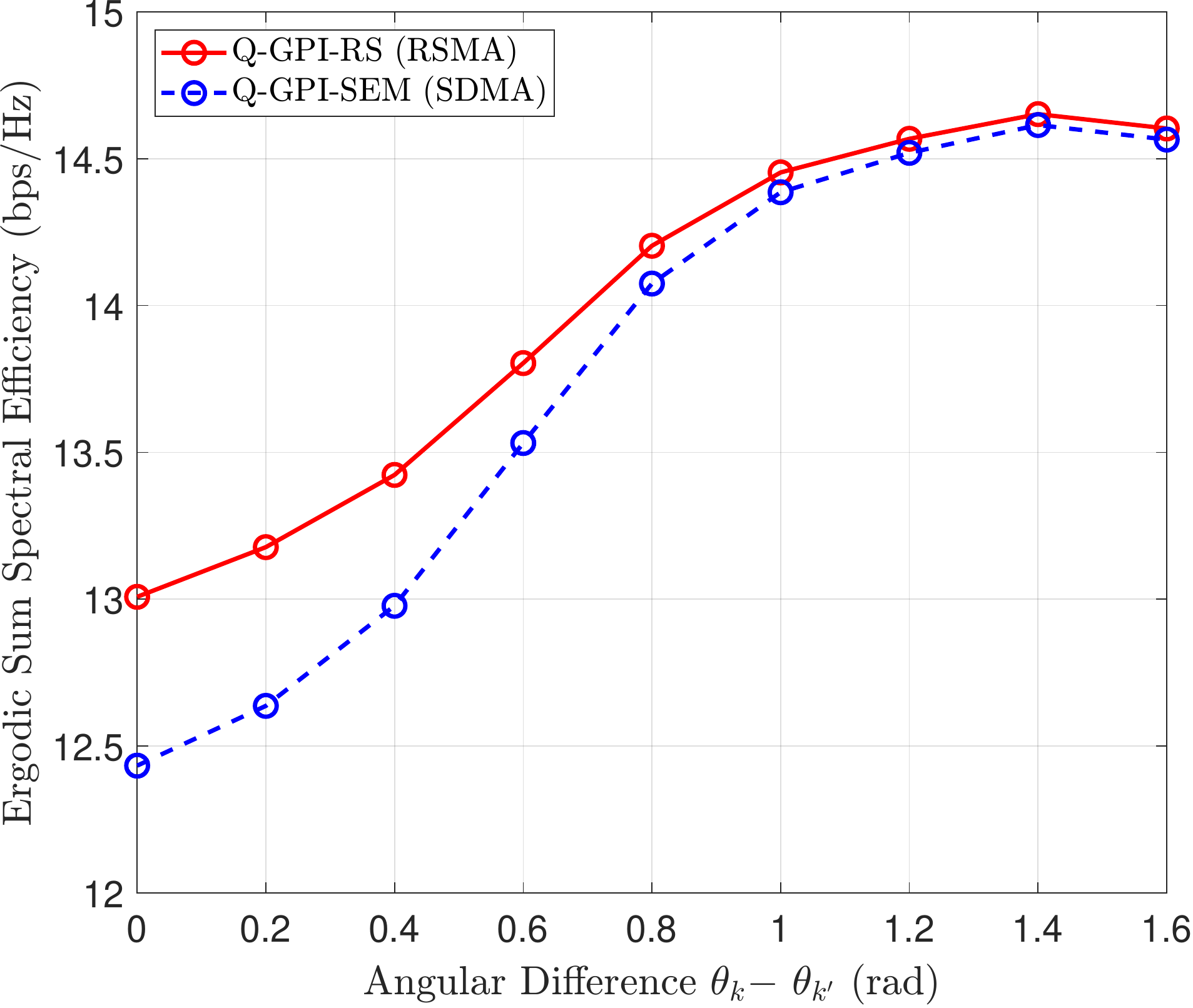}}}
	 	\vspace{-0.5em}
	\caption{The  sum spectral efficiency versus angular difference of users for $N = 4$ AP antennas, $ K = 2$ users, $b_{{\sf DAC},n} = 4, \forall n$, $b_{{\sf ADC},k} = 8, \forall k$, and SNR$ = 50$ dB.} 
 	\label{fig:N4K2AngularDAC4ADC8}
 	 	\vspace{-1em}
\end{figure}
In Fig.~\ref{fig:N4K2AngularDAC4ADC8}, we evaluate the spectral efficiency according to the user's angular difference $(\theta_k - \theta_{k'})$ for $N = 4$, $K = 2$, $b_{{\sf DAC},n} = 4,\ \forall n$, and $b_{{\sf ADC},k} = 8,\ \forall k$.
As shown in Fig.~\ref{fig:N4K2AngularDAC4ADC8}, the gap between Q-GPI-RS and the Q-GPI-SEM increases as the angular difference decreases.
In other words, the gain of RSMA increases as  user channel correlation increases, which was analytically shown for the perfect quantization system  in \cite{clerckx2019rate}.
This phenomenon occurs because, for highly correlated users, RSMA exploits the common stream 
to reduce the inter-user interference which is more severe when channel correlation is high.
SDMA, however, suffers from high inter-user interference.
User channels are typically correlated with small angular differences among the users in the indoor communications \cite{zafari2019survey} which are often the environment of low-power applications such as IoT communications.
Accordingly,
we assume the correlated channels where the angular difference between users is less than $\pi/6$  in the rest of the simulations.

\subsection{Heterogeneous Quantization Resolution}

\begin{figure}[!t]\centering
	\subfigure{\resizebox{0.42\columnwidth}{!}{\includegraphics{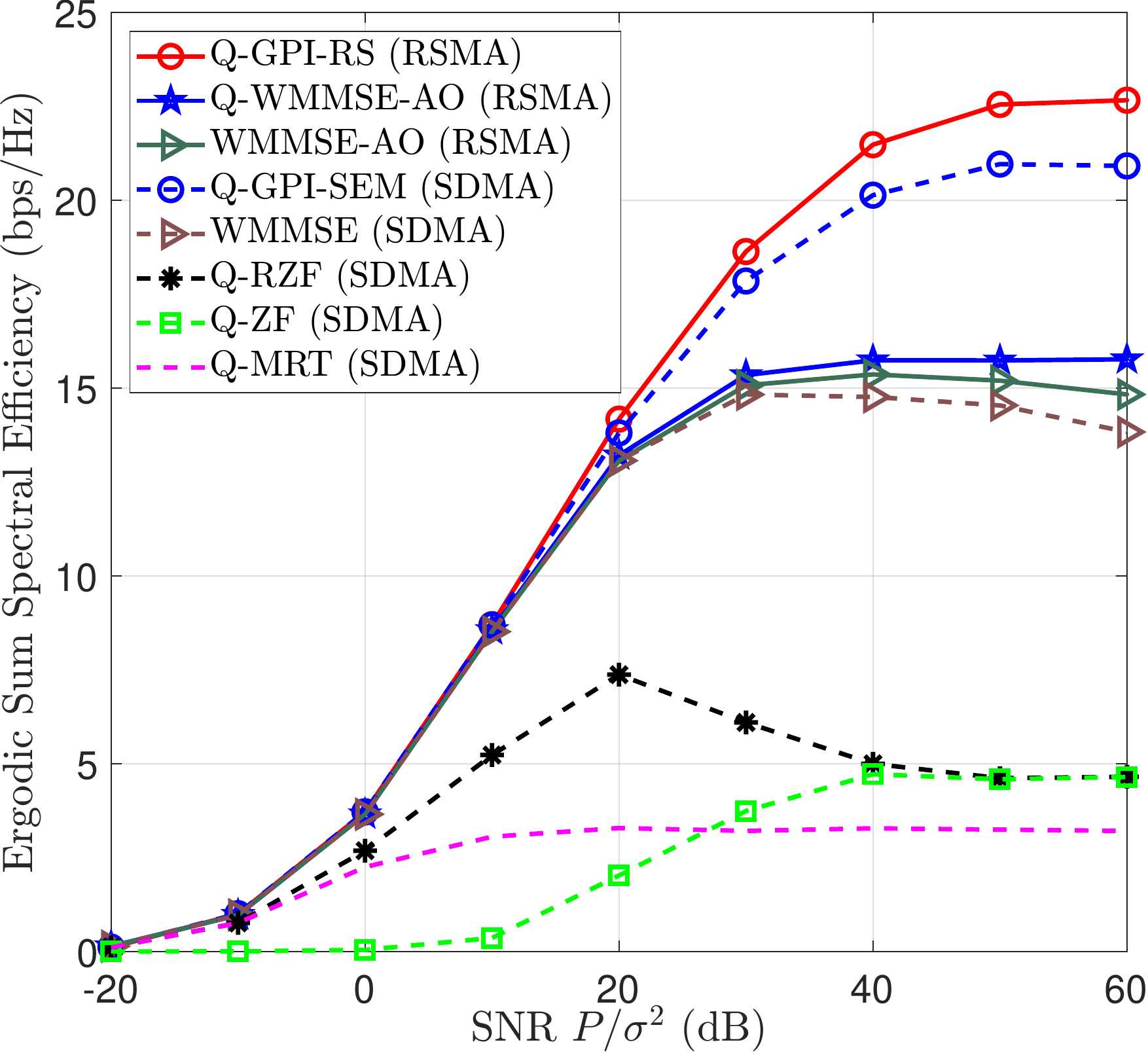}}}
	 	\vspace{-0.5em}
	\caption{The  sum spectral efficiency versus SNR for $N = 6$ AP antennas, $ K = 4$ users,  the numbers of DACs bits are set uniformly randomly from 2 to 8 bits, and $b_{{\sf ADC},k} = 8, \forall k$.} 
 	\label{fig:N6K4coCHSumSEheteroDACADC8}
 	 	\vspace{-1em}
\end{figure}

We analyze the spectral efficiency when DACs at the AP are configured with heterogeneous quantization bits.
In Fig.~\ref{fig:N6K4coCHSumSEheteroDACADC8}, for $N = 6$, $K = 4$, and $b_{{\sf ADC},k} = 8,\ \forall k$ we consider that the number of DAC bits is set uniformly randomly from 2 to 8 bits.
As shown in Fig.~\ref{fig:N6K4coCHSumSEheteroDACADC8}, Q-GPI-RS achieves the highest spectral efficiency than the other baseline methods with greater improvement compared to the previous random channel case.
In addition, the spectral efficiency of Q-WMMSE-AO is higher than that of WMMSE-AO by incorporating the quantization error.
We also note that WMMSE-AO decreases in the high SNR regime since WMMSE-AO method is developed without considering the quantization error.
Comparing RSMA with SDMA, Q-GPI-RS reveals the gain of RSMA from Q-GPI-SEM since the proposed method has better capability in reducing the inter-user interference by employing the common stream as the SNR increases.
The other RSMA algorithms, i.e., Q-WMMSE-AO  and WMMSE-AO, also provide higher spectral efficiency than WMMSE and the other conventional SDMA-based quantization-aware linear precoders because of the higher channel correlation, which is not the case in Fig.~\ref{fig:N4K2RanCHSumSEDAC4ADC6}.
Therefore, for randomly distributed heterogeneous DAC bits with correlated channels, the proposed Q-GPI-RS achieves the highest performance, and  RSMA  provides a noticeable gain over SDMA.

\begin{figure}[!t]\centering
	\begin{subfigure}[Sum spectral efficiency ]{\resizebox{0.45\columnwidth}{!}{\includegraphics{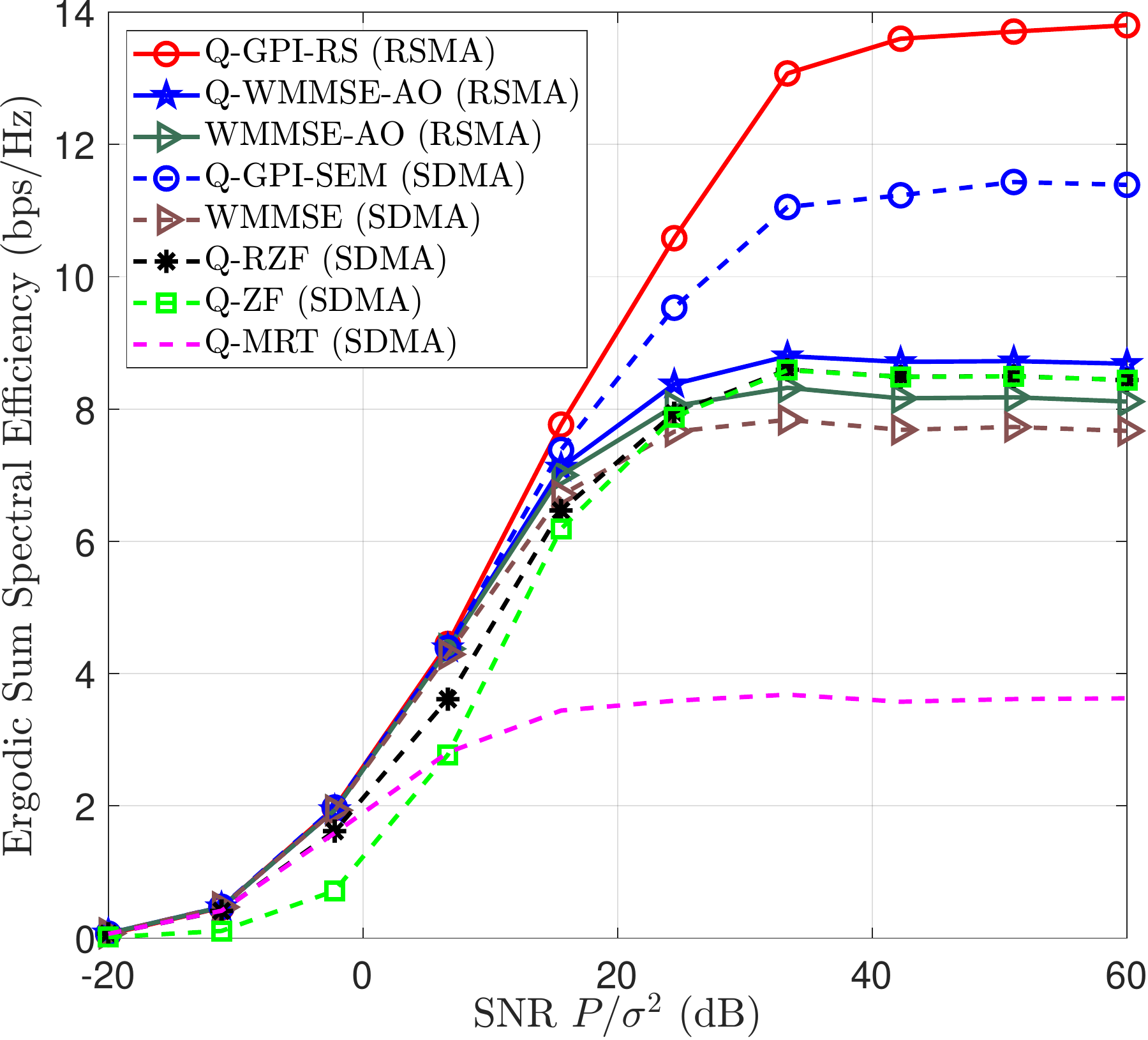}}}
	\end{subfigure}
	\begin{subfigure}[Spectral efficiency of Q-GPI-RS and Q-GPI-SEM]{\resizebox{0.45\columnwidth}{!}{\includegraphics{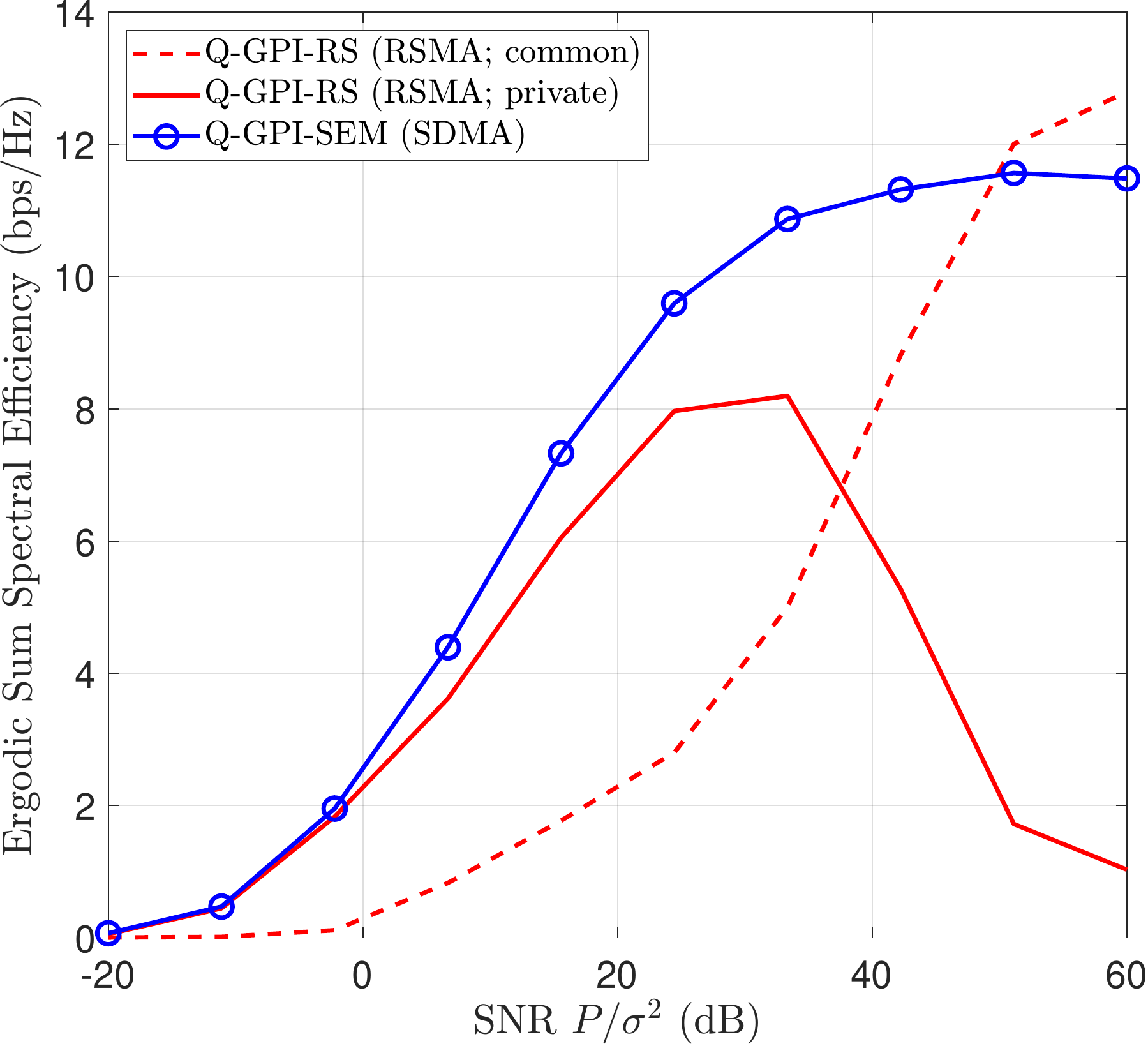}}}
	\end{subfigure}
	\begin{subfigure}[Antenna power ratios]{\resizebox{0.5\columnwidth}{!}{\includegraphics{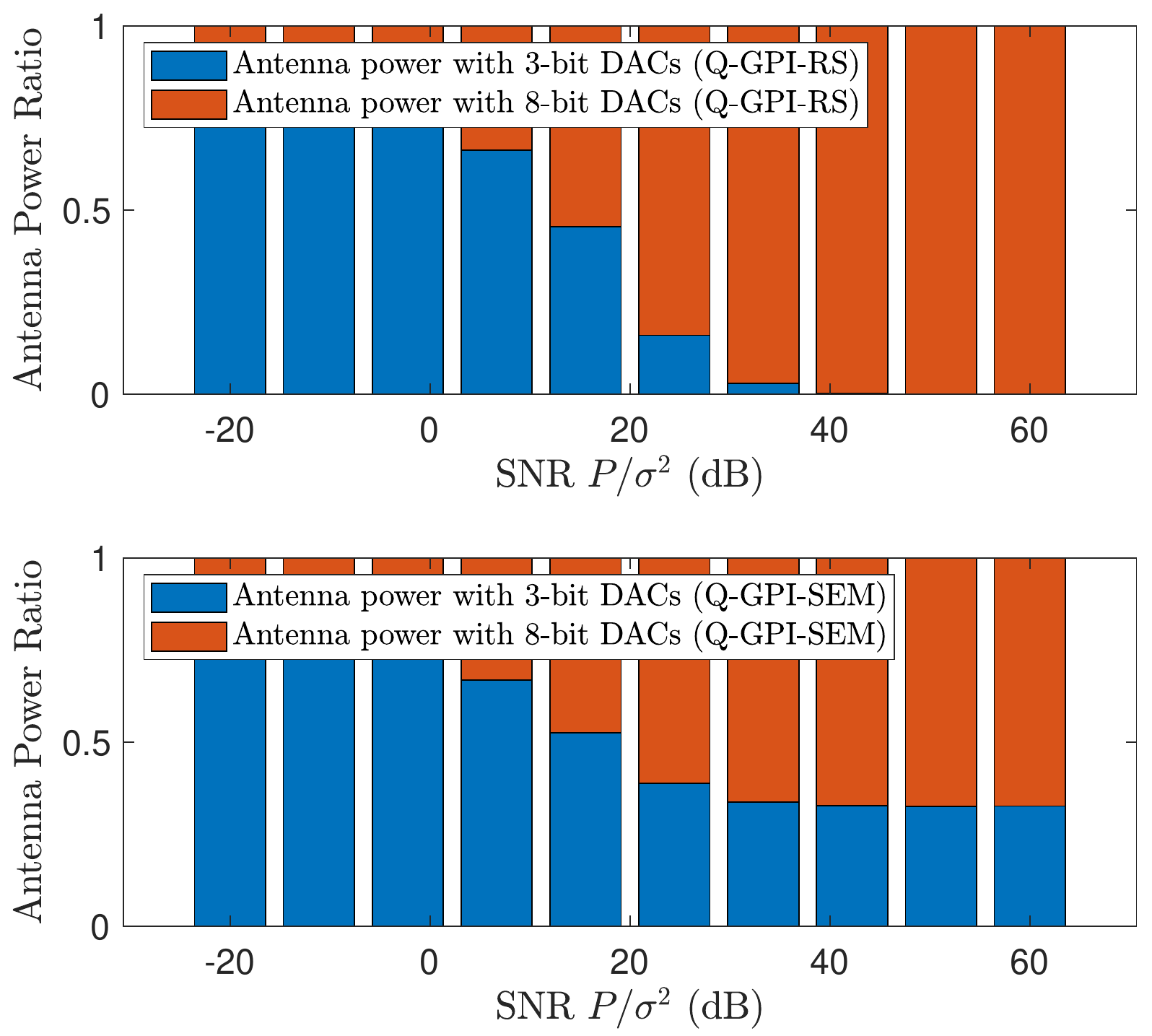}}}
    \end{subfigure}
	\caption{(a) The sum spectral efficiency, (b) the spectral efficiency of the common and private streams, and (c) the average antenna power ratio of Q-GPI-RS and Q-GPI-SEM versus SNR for $N = 4$ AP antennas, $ K = 2$ users, three DACs are equipped with $3$ bits and the other is equipped with $8$ bits, and $b_{{\sf ADC},k} = 8, \forall k$.} 
 	\label{fig:RS-SE-Hetero}
 	 	\vspace{-1em}
\end{figure}

In Fig.~\ref{fig:RS-SE-Hetero},  we consider the mixed DAC case where three DACs are equipped with $3$ bits and the other is equipped with $8$ bits for $N = 4$, $K = 2$, and $b_{{\sf ADC},k} = 8,\ \forall k$.
As shown in Fig.~\ref{fig:RS-SE-Hetero}(a), we note that Q-GPI-RS achieves the significant spectral efficiency improvement in medium-to-high SNR.
In particular, the gain of RSMA of Q-GPI-RS compared to  Q-GPI-SEM is achieved by assigning the high rate for the common stream as shown in Fig.~\ref{fig:RS-SE-Hetero}(b).
In addition, Fig.~\ref{fig:RS-SE-Hetero}(c) illustrates that  Q-GPI-RS concentrates the available transmit power on the antenna  with the high-resolution DAC (8 bits in this case) to prevent  the quantization error from increasing significantly as the SNR increases.
This trend leads to reducing the number of active antennas, and thus, the system can become overloaded.
In such a case,  RSMA transmission plays its key role by leveraging the common stream to deal with the overloaded system.
Q-GPI-SEM, on the other hand, allocates the transmit power to all antennas as shown in Fig.~\ref{fig:RS-SE-Hetero}(c). 
Accordingly, although  Q-GPI-SEM achieves the higher spectral efficiency than the other baselines, it still has performance limitation compared to Q-GPI-RS since  Q-GPI-SEM cannot avoid the large increase of quantization errors with the SNR. 
In this regard, it is concluded that RSMA is also highly beneficial in multiuser MIMO systems with mixed DACs.
\begin{figure}[!t]\centering
	\subfigure{\resizebox{0.5\columnwidth}{!}{\includegraphics{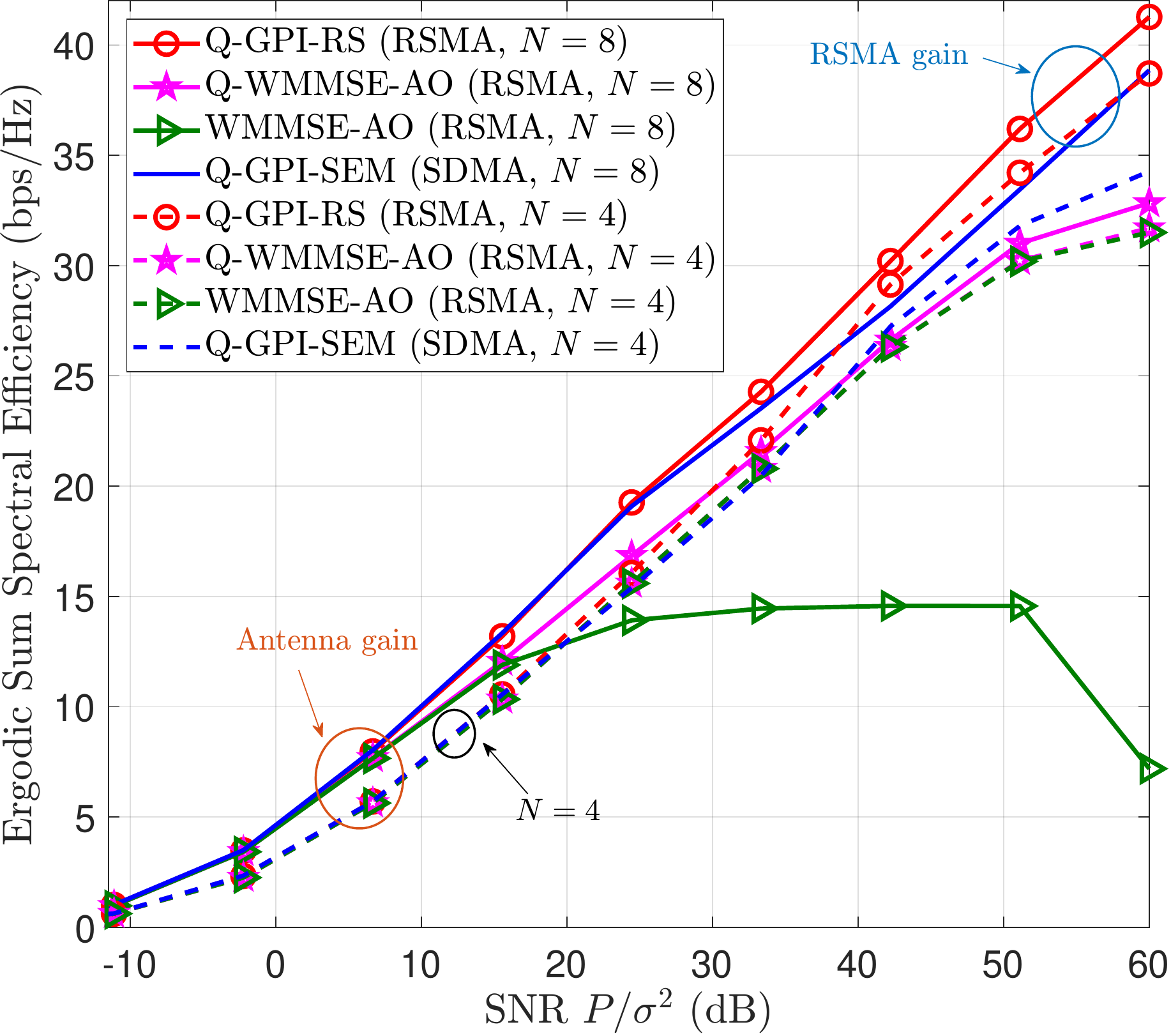}}} 	\vspace{-0.5em}
	\caption{The  sum spectral efficiency versus SNR for $N = 4$ AP antennas, $ K = 4$ users, $b_{{\sf DAC},n} = 10, \forall n$, and $b_{{\sf ADC},k} = 10, \forall k$ and for $N = 8$ AP antennas, $ K = 4$ users, half of the DACs are equipped with $3$ bits and the other half are equipped with $10$ bits, and $b_{{\sf ADC},k} = 10, \forall k$.} 
 	\label{fig:N8K4Comparison}
 	 	\vspace{-0.5em}
\end{figure}
Since the mixed ADC/DAC architecture which is a special case of the heterogeneous ADC/DAC systems has shown to offer high potential in maximizing the spectral efficiency with low power consumption in the MIMO systems \cite{zhang2018mixed}, we further evaluate the spectral efficiency versus the SNR for mixed DACs in which half of the DACs are equipped with $3$ bits and the other half are equipped with $10$ bits, $b_{{\sf ADC},k} = 10,\  \forall k$, $N = 8$, and $K = 4$.
In addition to the mixed DAC case, we plot the high-resolution DAC case for $N = 4$, $K = 4$, $b_{{\sf DAC},n} = 10,\ \forall n$, and $b_{{\sf ADC},k} = 10, \  \forall k$ for comparison.
As shown in Fig.~\ref{fig:N8K4Comparison}, the proposed Q-GPI-RS achieves the highest sum spectral efficiency in this environment. 
In particular, the gain of RSMA becomes larger as the SNR increases for both $N=8$ and $N=4$ cases.
We note that the spectral efficiency of $N=4$ case reduces its gap from that of $N=8$ case as the SNR increases.
This is due to the phenomenon that Q-GPI-RS with mixed-resolution DACs reduces the quantization error in the high SNR by allocating more power to antennas with high-resolution DACs.
Accordingly, in the high SNR, the effective number of antennas becomes the number of antennas with high-resolution DACs, which makes $N=8$ and $N=4$ cases similar, while increasing the gap from the SDMA-based approaches.
In the low-to-medium SNR, however, the gain of RSMA is small but there exists an antenna gain for $N=8$  compared to $N=4$   thanks to additional antennas with low-resolution DACs.
Therefore, using the proposed RSMA precoding method with mixed-resolution DACs offers antenna gain in the low-to-medium SNR with the small cost of deploying the low-power hardware, and also provides the  RSMA gain in the high SNR.

\begin{figure}[!t]\centering
	\subfigure{\resizebox{0.45\columnwidth}{!}{\includegraphics{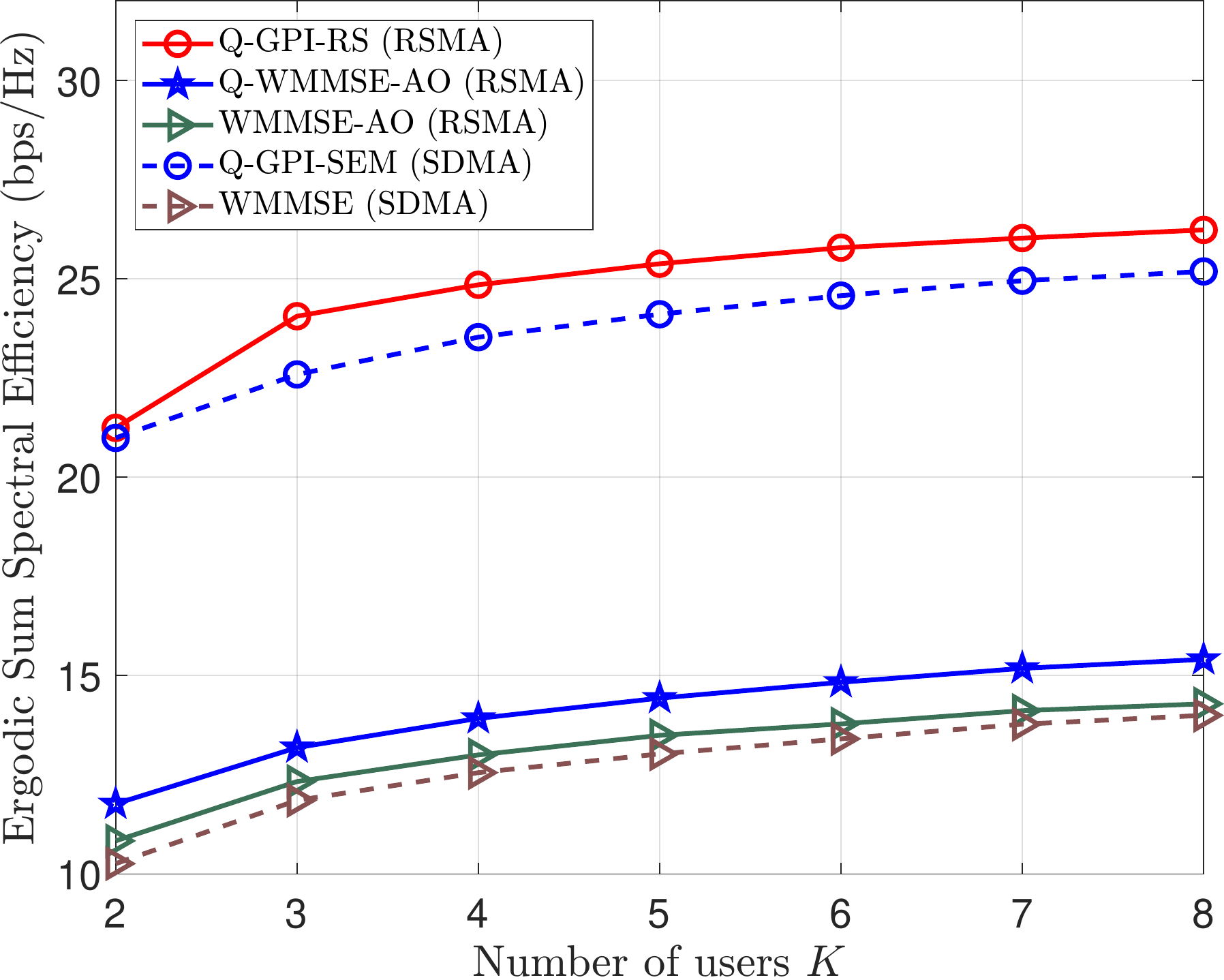}}}
	 	\vspace{-0.5em}
	\caption{The  sum spectral efficiency versus the number of users for $N = 8$ AP antennas,  the numbers of DAC bits are set uniformly randomly from 2 to 8 bits, and $b_{{\sf ADC},k} = 8, \forall k$, and SNR $= 40$ dB.} 
 	\label{fig:N8coCHSumSEforusersADC8}
 	 	\vspace{-1em}
\end{figure}

In Fig.~\ref{fig:N8coCHSumSEforusersADC8}, we evaluate the proposed algorithm with respect to the number of users for $N = 8$.
The number of  DAC bits is set uniformly randomly from 2 to 8 bits and $b_{{\sf ADC},k} = 8,\ \forall k$.
As shown in Fig.~\ref{fig:N8coCHSumSEforusersADC8}, the proposed Q-GPI-RS achieves the gain of RSMA compared to the Q-GPI-SEM as we increase the number of users.
The gain of RSMA is also observed when comparing WMMSE-AO with WMMSE.
The gain of RSMA seems to be similar over the considered number of users; the common rate is determined by the minimum rate of  RSMA users for the common stream, and thus, the gain of  RSMA cannot arbitrarily increase with the number of users.
Noticing that the other baseline methods, namely, Q-WMMSE-AO and WMMSE-AO, show large gaps from the proposed method even though they use RSMA transmission, we can conclude that the proposed method is able to properly design precoders for both the common stream and private streams by incorporating the quantization errors.


\begin{figure}[!t]\centering
	\begin{subfigure}[DAC resolution]{\resizebox{0.42\columnwidth}{!}{\includegraphics{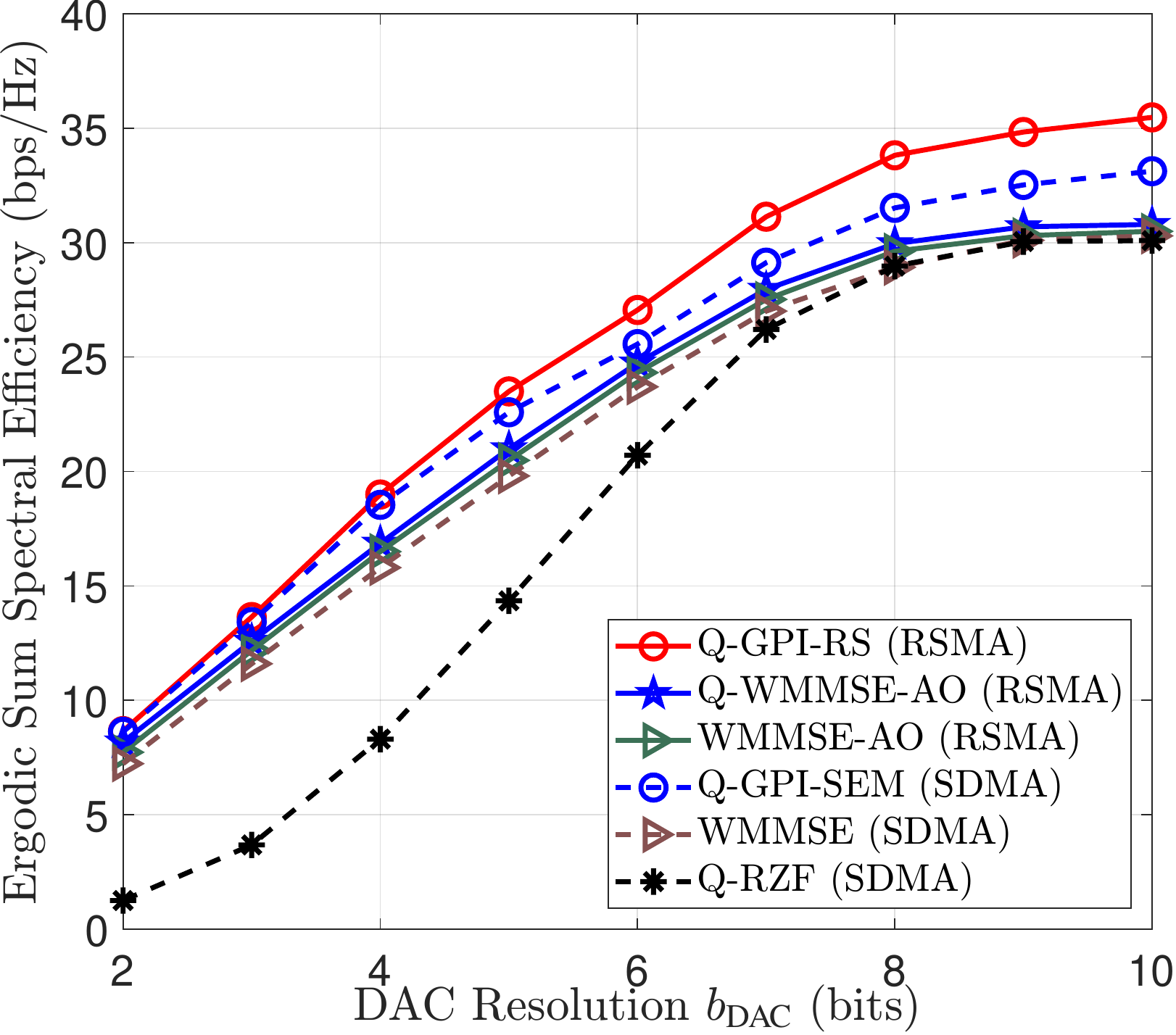}}}
	\end{subfigure}
	\begin{subfigure}[ADC resolution]{\resizebox{0.42\columnwidth}{!}{\includegraphics{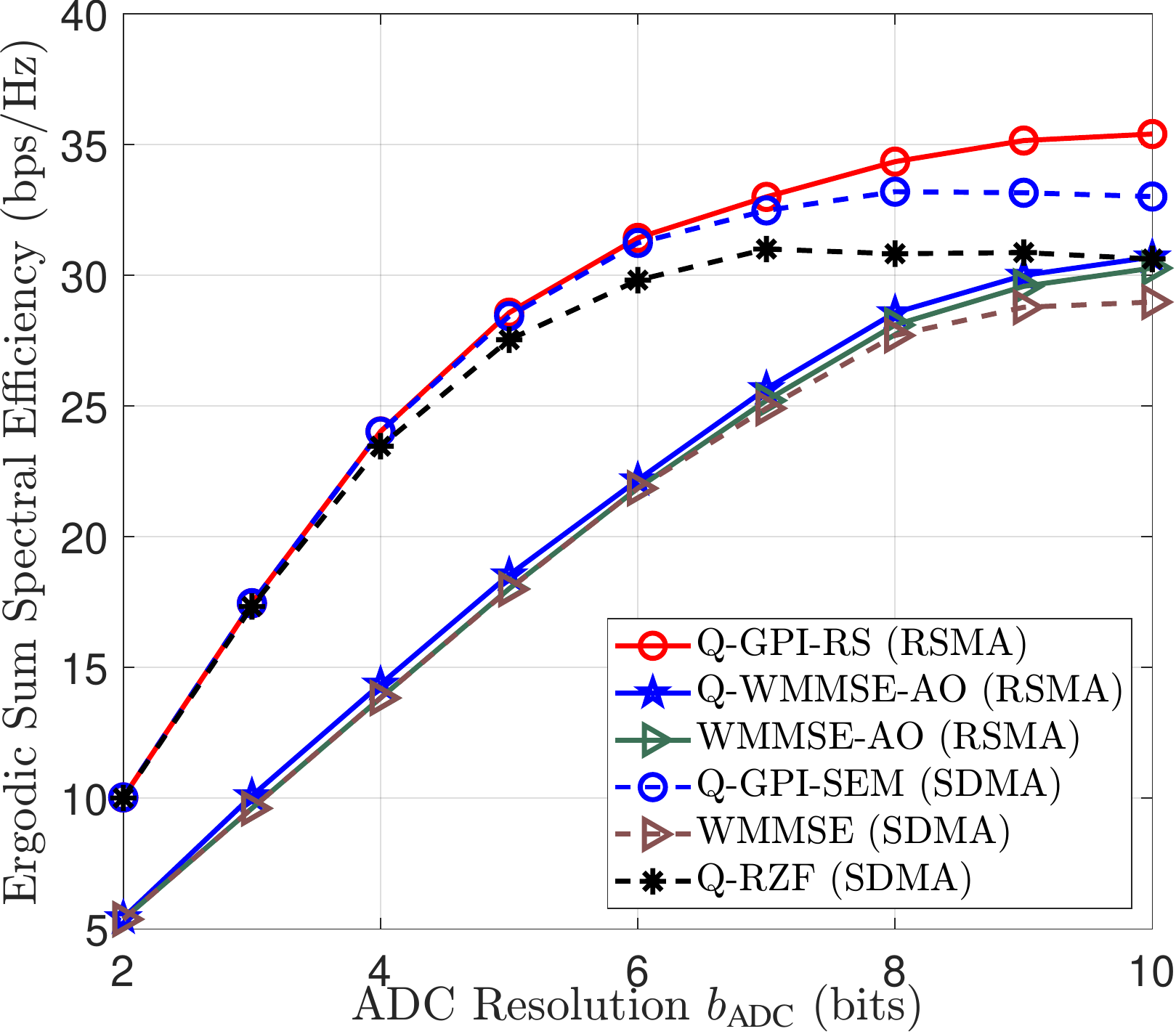}}}
	\end{subfigure}
	\caption{
	(a) The  sum spectral efficiency versus the number of DAC bits for $N = 8$ AP antennas, $ K = 4$ users,  $b_{{\sf ADC},k} = 10, \forall k$, and SNR $= 40$ dB and (b) the  sum spectral efficiency versus the number of  ADC bits for $N = 8$ AP antennas, $K = 4$ users,  $b_{{\sf DAC},n} = 10, \forall n$, and SNR $= 40$ dB.}
 	\label{fig:RS-SE-DAC-ADC}
 	 	\vspace{-1em}
\end{figure}

\subsection{Effect of DAC and ADC Quantization to RSMA}
\label{subsec:DACADCeffect}
Now, we analyze the effect of the number of DAC and ADC bits to verify the conjecture in Remark~\ref{rm:effect}.
In Fig.~\ref{fig:RS-SE-DAC-ADC}(a), we plot the spectral efficiency versus the number of DAC bits for $N = 8$, $K = 4$, and $b_{{\sf ADC},k} = 10,\ \forall k$.
We also evaluate the spectral efficiency versus the number of ADC bits for $N = 8$, $K = 4$, and $b_{{\sf DAC},n} = 10,\ \forall n$ in Fig.~\ref{fig:RS-SE-DAC-ADC}(b).
As shown in Fig.~\ref{fig:RS-SE-DAC-ADC}, the proposed Q-GPI-RS achieves the highest performance over the considered number of bits.
The gain of RSMA becomes larger as DAC and ADC resolutions increase because RSMA can allocate the higher rate to the common stream with the higher resolution, which corresponds to the observation in \cite{ahiadormey2021performance}.
In addition, the gain of RSMA is more sensitive to the ADC resolution than the DAC resolution since the rate of the common stream is more dominated by the number of ADC bits than the number of DAC bits, which confirms the analysis in Remark~\ref{rm:effect}.
It is concluded that the medium-to-high resolution DACs and high-resolution ADCs are required to fully realize the benefit of RSMA in the homogeneous ADC and DAC cases.

\section{Conclusions}
In this paper, we proposed a promising precoding algorithm for downlink RSMA systems with low-resolution quantizers to maximize the sum spectral efficiency.
Since the formulated spectral efficiency maximization problem is non-smooth,
we used the LogSumExp technique to convert the problem into a tractable form.
In addition, we further reformulated the non-convex spectral efficiency maximization problem to the product of the Rayleigh quotients for more tractability.
We then derived the first-order optimality condition to find the stationary point.
Interpreting the condition as the generalized eigenvalue problem, we developed a computationally efficient algorithm to find the best stationary point that maximizes the spectral efficiency.
Simulation results demonstrated that the proposed method achieved  the highest sum spectral efficiency compared to the baseline methods.
It was also observed that the gain of RSMA exists in most DAC and ADC resolutions.
In particular, the gain of RSMA increased with the resolutions of DACs and ADCs because the common stream can have a higher rate with the higher resolutions.
Furthermore, Q-GPI-RS attained a significant gain of RSMA in heterogeneous DACs by decreasing the allocated transmit power for the antennas with low-resolution DACs as a means of  suppressing the quantization error, which shows the high flexibility of RSMA in using antennas with different DAC resolutions.
Based on the observations, we confirmed that RSMA benefits coarse quantization systems and that the proposed RSMA precoding algorithm significantly improves the spectral efficiency, offering high adaptation to the heterogeneous quantization bits thanks to RSMA.
Therefore, the proposed algorithm can provide benefits in both increasing the communication efficiency and designing low-power transceiver architectures for future wireless communication systems.
\appendices
\section{Proof of Lemma \ref{lem:main}} \label{proof:lem1}
From the reformulated optimization problem, the Lagrangian function is defined as
\begin{align}
    \label{Largaragianfuction}
    L(\bar {\bf{w}})&=\ln \left(\frac{1}{K}  \sum_{k = 1}^{K}  \exp \left(\log_2 \left( \frac{\bar {\bf{w}}^{\sf H} {\bf{A}}_{{\sf c},k} \bar {\bf{w}}}{\bar {\bf{w}}^{\sf H} {\bf{B}}_{{\sf c},k}  \bar {\bf{w}} }  \right)^{-\frac{1}{\tau}} \right)  \right)^{-\tau}\!+ 
    \sum_{k = 1}^{K}  \frac{1}{\ln2}\ln \left( \frac{\bar {\bf{w}}^{\sf H} {\bf{A}}_k \bar {\bf{w}}}{\bar {\bf{w}}^{\sf H} {\bf{B}}_k \bar {\bf{w}}}\right).
\end{align}
Then, we compute partial derivatives of the Lagrangian function \eqref{Largaragianfuction} and find the stationarity condition by setting \eqref{Largaragianfuction} to zero.
Let us denote the first term and second term in \eqref{Largaragianfuction} as $L_1(\bar\bw)$ and $L_2(\bar\bw)$, respectively.
The partial derivative of $L_1(\bar {\bf{w}})$ is represented as
\begin{align}
    \label{eq:L1_d}
    \frac{\partial L_1(\bar {\bf{w}})}{\partial \bar     {\bf{w}}^{\sf H}}
    = \frac{1}{\ln 2} \sum_{k = 1}^{K}   \Bigg[ \frac{\exp\left( -\frac{1}{\tau}  \log_2\left(\frac{\bar {\bf{w}}^{\sf H} {\bf{A}}_{{\sf c},k} \bar {\bf{w}}}{\bar {\bf{w}}^{\sf H} {\bf{B}}_{{\sf c},k} \bar {\bf{w}} } \right) \right)}{\sum_{\ell = 1}^{K} \exp\left( -\frac{1}{\tau} \log_2\left(\frac{\bar {\bf{w}}^{\sf H} {\bf{A}}_{{\sf c},\ell} \bar {\bf{w}}}{\bar {\bf{w}}^{\sf H} {\bf{B}}_{{\sf c},\ell} \bar {\bf{w}} } \right)\right)} \left\{\frac{{\bf{A}}_{{\sf c},k} \bar {\bf{w}}}{\bar {\bf{w}}^{\sf H} {\bf{A}}_{{\sf c},k} \bar {\bf{w}}} - \frac{{\bf{B}}_{{\sf c},k} \bar {\bf{w}}}{\bar {\bf{w}}^{\sf H} {\bf{B}}_{{\sf c},k} \bar {\bf{w}}} \right\}  \Bigg].
\end{align}
The derivative of the Lagrangian function $L_2(\bar {\bf{w}})$ is computed as
\begin{align} 
    \label{eq:L2_d}
    \frac{\partial L_2 (\bar {\bf{w}})}{\partial \bar {\bf{w}}^{\sf H}} 
    =&\frac{1}{\ln 2} \sum_{k = 1}^{K} \left[\frac{{\bf{A}}_k \bar {\bf{w}}}{\bar {\bf{w}}^{\sf H} {\bf{A}}_k \bar {\bf{w}}} - \frac{{\bf{B}}_k \bar {\bf{w}}}{\bar {\bf{w}}^{\sf H} {\bf{B}}_k \bar {\bf{w}}} \right].
\end{align}
Using  \eqref{eq:L1_d} and \eqref{eq:L2_d}, the first-order optimality condition is given as
\begin{align}
    \label{eq:kkt_inproof}
    &\frac{\partial L_1(\bar {\bf{w}})}{\partial \bar {\bf{w}}^{\sf H}} + \frac{\partial L_2(\bar {\bf{w}})}{\partial \bar {\bf{w}}^{\sf H}}  
    \\ 
    \label{eq:L_d}
    & = 
    \sum_{k = 1}^{K}   \Bigg[\! \frac{\exp\left( -\frac{1}{\tau}  \log_2\left(\frac{\bar {\bf{w}}^{\sf H} {\bf{A}}_{{\sf c},k} \bar {\bf{w}}}{\bar {\bf{w}}^{\sf H} {\bf{B}}_{{\sf c},k} \bar {\bf{w}} } \right) \right)}{\sum_{\ell = 1}^{K} \exp\left( -\frac{1}{\tau} \log_2\left(\frac{\bar {\bf{w}}^{\sf H} {\bf{A}}_{{\sf c},\ell} \bar {\bf{w}}}{\bar {\bf{w}}^{\sf H} {\bf{B}}_{{\sf c},\ell} \bar {\bf{w}} } \right)\right)} \left\{\frac{{\bf{A}}_{{\sf c},k} \bar {\bf{w}}}{\bar {\bf{w}}^{\sf H} {\bf{A}}_{{\sf c},k} \bar {\bf{w}}}\! -\! \frac{{\bf{B}}_{{\sf c},k} \bar {\bf{w}}}{\bar {\bf{w}}^{\sf H} {\bf{B}}_{{\sf c},k} \bar {\bf{w}}}\! \right\}  \!\!\Bigg] + \sum_{k = 1}^{K} \left[\frac{{\bf{A}}_k \bar {\bf{w}}}{\bar {\bf{w}}^{\sf H} {\bf{A}}_k \bar {\bf{w}}} - \frac{{\bf{B}}_k \bar {\bf{w}}}{\bar {\bf{w}}^{\sf H} {\bf{B}}_k \bar {\bf{w}}} \right] 
    \\
    &= 0.
\end{align}
Rearranging \eqref{eq:L_d}, we derive ${\bf{A}}_{\sf KKT}(\bar {\bf{w}}) \bar {\bf{w}} = \lambda(\bar {\bf{w}}) {\bf{B}}_{\sf KKT} (\bar {\bf{w}}) \bar {\bf{w}}$.
Since ${\bf B}_{{\sf c},k}$ and ${\bf B}_{k}$ are Hermitian block diagonal matrices, ${\bf B}_{\sf KKT}$ is invertible.
Accordingly, we finally obtain the condition in \eqref{eq:lem_kkt_stack}.
This completes the proof.
\qed

\bibliographystyle{IEEEtran}
\bibliography{draft_final}

\end{document}